\RequirePackage{fix-cm}
\documentclass{svjour3}                     
\smartqed
\usepackage{graphicx}
\usepackage{mathptmx}      
\usepackage[colorlinks,citecolor=blue]{hyperref}
\usepackage[T1]{fontenc}
\usepackage{cite}
\usepackage{amsmath}
\usepackage{amsfonts}
\usepackage{array}
\usepackage[caption=false,font=footnotesize]{subfig}
\usepackage{amsmath}
\usepackage{xfrac}
\usepackage{color}

\newcommand{\e}{\textrm{e}}
\newcommand{\A}{\textrm{A }}
\newcommand{\B}{\textrm{B }}
\newcommand{\R}{\textrm{R }}

\journalname{Wireless Networks} 

\begin{document}	
	\title{Effective Capacity Maximization of Two-Way Full-Duplex and Half-Duplex relays with Finite Block Length Packets Transmission}		
	\titlerunning{Effective Capacity Maximization of Two Way Full-Duplex and}
	
	\author{Mohammad Lari \and Zahra Keshavarz Gandomani}
	\institute{
		M. Lari \at
		Electrical and Computer Engineering Faculty, Semnan University, Semnan, Iran \\
		\email{m\_lari@semnan.ac.ir}
		\and
		Zahra Keshavarz Gandomani \at
		Electrical and Computer Engineering Faculty, Semnan University, Semnan, Iran \\
		\email{zahra\_keshavarz@semnan.ac.ir}
		}
	
	\date{Received: date / Accepted: date}	
	\maketitle
	
	\begin{abstract}
	In order to satisfy the delay requirements of telecommunication systems, in this paper, we present a cooperative network with the short packet transmission in the Rayleigh fading channel.	The desired relay can be implemented as a two-way half-duplex (HD) or a two-way full-duplex (FD). Also, for more accurate satisfaction and reduction of communication delays, sending and receiving with short packets is considered. Effective capacity appropriately measures the transmission rate under the delay constraint. Therefore, it is considered as a performance evaluation criterion here. With a two-way relay, two nodes exchange data with each other using a relay simultaneously. The priorities and requirements of the two nodes are not necessarily the same. Therefore, to increase performance, the system is modeled and solved as a multi-objective problem. In this way, the available power in the network is divided between the relay and two nodes, and the effective capacity of the two nodes is maximized. Depending on the different conditions, the optimal amount of allocated power to relay and nodes is calculated. However, due to the complexity and time consuming calculations, an approximate method which speeds up the calculation is presented. The approximated solution has a very close performance to the optimal allocated power. Finally, various comparisons have been made in different conditions between the performance of two-way HD and two-way FD relays. The improvement of multi-objective power allocation has been shown, especially when the relay is not located in the middle of two nodes. 
		
	\end{abstract}


\section{Introduction}\label{sec:introduction}
Fifth generation of telecommunication systems brings a new communication era, which would transform several industry aspects and people's lives. The Internet of Things (IoT) \cite{p1,p2}, self-driving and autonomous vehicles \cite{p3}, online control and industrial automation \cite{p4}, and virtual reality \cite{p5,p6} are instances of the fifth generation technology applications. For example, the IoT is an extensive concept, which denotes that all our surrounding objects are interconnected. This concept is applicable in all areas of energy, transportation, health, manufacturing, production, etc. to create an intelligent world with optimal use of resources \cite{p1,p2,p7}. To respond to the new requirements discussed and targeted in the fifth generation of communications, the network structure and its various parameters will undergo a fundamental transformation. Consequently, it will be possible to provide new services and meet different needs in advance.

The delay parameter and its association with delay-sensitive traffic are ignored or underestimated in previous generations \cite{p8,p9}. Human being has been pivotal in designing the previous generations of telecommunication systems (first to the fourth generation). As the human senses are less sensitive to low latencies, the delay parameter is usually not guaranteed in previous generations and tens of milliseconds delay were reasonably normal \cite{p10}. These circumstances are not tolerable for new requirements such as self-driving cars or remote surgery. In these applications and many other situations (the details of which are provided in various sources such as \cite{p7,p11,p12}), considering even a few milliseconds delay and guaranteeing its level of latency are compulsory \cite{p10}. In such cases, if the delay exceeds the threshold, the packet will be unusable and is usually discarded even if received completely and without any error. If a self-driving car transmits data to the approaching car or when overtaking a car with delay, the result will be catastrophic. Therefore, paying attention to delay at different network layers and transmission-reception techniques are focused recently. So, this paper will guarantee the statistical delay of packet by proposing a cooperative model for data transmission in the shortest possible time and appropriate allocated power. 

Delay is an extensive concept created in a telecommunication systems or networks for various reasons \cite{p13,lari2019transmission}. Therefore, it is necessary to pay attention to its origin to reduce delay and guarantee its maximum level. One reason for the delay is to transmit and decode relatively long packets in the receiver. Since it is possible to achieve a rate equal to the channel capacity for (infinitely) long packets in Shannon's theorem, the packets transmitted from the transmitter to the receiver have been long in recent years and conventional telecommunications. If these long packets are designed appropriately, (like using the Turbo codes or Low-Density Parity Check (LDPC) codes) data transmission rate will be very close to Shannon's rate. However, transmitting these packets is time-consuming, and their decoding delay is high. Therefore, they are not applicable in delay-sensitive traffics. On the other hand, in many applications, such as sensor networks or the Internet of Things, the data generated by the transmitters is very low and does not exceed a few bytes and it is not efficient to use long packets to transmit such data\cite{p14,p15}. Therefore, transferring the Finite Block Length (FBL) packets for data transmission and analyzing the performance of the telecommunication system by sending FBL packets are highly considered \cite{p16}. The conducted studies indicate that the communication rate differs from Shannon's rate when using the FBL packets and the error of the packets in the receiver is not zero. However, the transmission and decoding time of these packets is lower; therefore, this type of communication is being analyzed and implemented in the fifth generation of telecommunication systems \cite{p16}. Besides the transmission and decoding delays, the transmitter buffer may cause delay. In addition, the limited-time and frequency intervals in the telecommunication system could also cause a delay in packet transmission. Several other factors are also available in telecommunication networks, which might cause delay but they are not discussed here (for further information see \cite{p13}).

To provide high-quality service to delay-sensitive traffics, in this paper, a cooperative system including two nodes and a half-duplex (HD) two-way, or a full-duplex (FD) two-way relay between two nodes are considered. For ease of understanding of the content, the half-duplex and full-duplex performance is abbreviated as HD and FD, respectively. Using a simple relaying system, four time slots or frequencies are required to transfer data between two nodes. However, if a two-way relay is applied, the number of time or frequency slots is reduced. Two-way relaying can be applied for direct connection between devices without any interface (i.e. device-to-device connection or D2D) in a telecommunication network \cite{p17}. In a two-way relay with HD function, the relay receives data from both nodes within a time slot (or frequency) and retransmits the received data to both nodes in the next interval. Therefore, the number of required time intervals (or frequencies) is equal to 2. Similarly, two-way relay with FD function receives data from nodes and transmits data back to them in the same interval. Therefore, only 1 time (or frequency) interval is required. In this case, the nodes should also act as FD. Thus, the relays and nodes face the self-interference problem i.e. power transmission leakage to the self-receiver that should be managed and reduced properly. Accordingly, besides increasing the spectral gain, the use of two-way relays reduces delay \cite{p18,p19}.

The transmitter buffer is another bottleneck that causes delay. If the packet input and output rate to the buffer is not adjusted correctly, the suppression of data in the transmitter buffer will increase and this will cause a delay. Since the packet output rate from the transmitter buffer is proportional to the random capacity and varies with the wireless channel, the proper adjustment of the packet input rate to the buffer is very important. If the input rate is too much, packet latency and delay in the buffer will be high. On the other hand, if the input rate is too low, it will not be possible to provide high-speed services. This input rate, which is adjusted to the requested service quality and guarantees the statistical delay of packets latency in the buffer, is called the effective capacity. To maximize the effective capacity of two nodes in this paper, the optimal power is calculated for allocation to relay and nodes. Since the two nodes are different and their requested rate and quality of service and constraints are not necessarily identical, the effective capacity of both nodes must be maximized. Therefore, a multi-objective optimization problem arises. This problem will be discussed below and resolved appropriately. Then, the effective capacity of the system will be discussed under different circumstances based on system parameters like the residual self-interference, the received small packet error, and service quality parameter. Compared to similar studies, the innovations of this research paper are as follows.

•	Optimizing the effective capacity of the FD relay and two-way relay is the main innovation of this article. Using the FD and two-way relays can reduce the data transmission delay between nodes and they are suitable to transmit the delay-sensitive data.

•	Another innovation of this paper is to provide a multi-objective model for the optimization problem and the detailed study and provision of the relationships associated with it. The nature of the problem is multi-objective. However, most previous articles and reviews have not optimized these systems on a multi-objective basis. Therefore, considering the problem as a multi-objective phenomenon provides interesting and different points.

•	Here, it is proved that the optimization problem's solution is time-consuming and has high computational complexity. Therefore, an approximate but simpler solution, which is very close to the optimal performance, is proposed that could contribute to solving various situations quickly in a shorter time. 

The rest of this paper is structured as follows. Section 2 provides an overview of the previous research related to the subject of the article. Section 3 describes the telecommunication of FBL packets, the two-way HD, and FD relays, effective capacity, and channel model. Section 4 examines and solves the multi-objective optimization problem to maximize the effective capacity of two nodes. Section 5 provides the results of numerous simulations and comparisons and finally the paper is concluded in Section 6. 


\section{Review of the previous studies}\label{sec:review}
Ultra-Reliable Low Latency Communications (URLLC) is one of the scenarios considered in the fifth generation of telecommunication systems. The amount of information transferred between transmitter and receiver is low in this type of communication. However, even this low information must be received with low delay and error to be applicable \cite{p13,p20}. Article \cite{p15} was the beginning of the research on this type of communication and references \cite{p21,p22,p23} were the continuation of this research trend. Since the quality of service in this type of communication is widely different from conventional communication (in conventional communication, the delay is not important, but transmission rate matters; however in URLLC, the transmission rate is not so important, but low delay and insignificant error matter), references \cite{p21,p22,p23} have addressed resource management to provide appropriate services to users in both categories. For example, \cite{p21} has considered a cellular network including different users with the various requested quality of services. One group is the conventional network users who require high rates and do not care about delay. The other category is the users who care about low delay and error rather than the transmission rate. The base station has to distribute time and frequency intervals and its power among users to maximize the capacity. In other cases, such as \cite{p24}, increasing the throughput rate of URLLC users and reducing their error under delay constraint are considered. For further reviews and other issues see \cite{p11,p13,p17}.

In recent years, as the telecommunication systems and user requirements became complex, the multi-objective optimization of telecommunication systems is focused \cite{p25,p26}. In this optimization method, several objective functions are optimized simultaneously rather than a single one. Thus, performance improvement compared to one-goal optimization is expected \cite{p27}. Therefore, in \cite{p21,p24} the problem is proposed and solved as multi-objective due to the presence of different objectives in URLLC.

FD communication and relays are emphasized due to their potential performance improvements \cite{p28}. Besides increasing the spectral efficiency, the use of FD relays reduces the transmission time between source and destination. Therefore, it is appropriate for delay-sensitive applications. For example, the cellular network users in \cite{p29} need low-latency service. FD relay is devised to provide the corresponding services to cell edge users. Self-interference reduction processes are Time-consuming and increases the delay. Therefore, the FD relay decreases the self-interference with the least possible processing and the residual interference is processed at the base station. Then, if necessary, this interference is used to reduce error. The authors in \cite{p29m} calculated and compared the reliability when using the FBLs in HD and FD relays. Then, they specified the threshold value for selecting the type of relay between HD and FD. Moreover, the FD state of base station performance is considered to provide the URLLC service to users \cite{p30m,p31m}. Like the FD relays, the two-way relays can increase the spectral efficiency and reduce the delay by declining the required time intervals when transferring data between two nodes \cite{p32m}.

Effective capacity is a recent definition of capacity, which considers the delay of storing data in the buffer. Effective capacity is a proper criterion to analyze the delay-sensitive users' performance and guaranteeing the quality of their requested service. One of the uses of low-latency transmission is telecommunication between vehicles. Therefore, the effective capacity is maximized by allocating power and bandwidth and considering the delay constraint in \cite{p33m}. In this paper, channel variations are rapid because of the high speed of vehicles; therefore, large-scale channel state information are applied to allocate power. Another application of low-latency transmission is the industrial application and Machine Type Communications (MTC). In \cite{p34m}, an industrial connection with some parallel fading channels is considered. The entry of transmitted packets is scattered and only the channel state information is statistically available in the transmitter. Accordingly, the effective capacity and quality of service are investigated and an upper bound is calculated for the possibility of packet error. In \cite{p35m,p36m}, the resource allocation to maximize the effective capacity of the Cloud Radio Access Network (C-RAN) is investigated for URLLC transmission. The authors of \cite{p37m} has addressed the performance of a telecommunication system with amplify and forward relay when transmitting the FBL packets. By allocating power in this paper, the effective capacity is maximized. Similar power allocation is addressed in \cite{p38m} to maximize the effective capacity of a downlink multi-user network when transmitting the FBL packets. However, there is no reliable article addressing the effective capacity of low-latency transmission in FD relays. Therefore, one of the innovations of this paper is to consider this model for the delay-sensitive data packets. 


\section{System Model}\label{sec:system_model}

Here, the cooperative system including two-way HD and FD relays with transmitting the FBL packets is addressed. Effective capacity is also explained as the data entry rate to the transmitter buffer, which also guarantees the statistical delay of packets. 

\subsection{Short Packet Transmission}\label{subsec:short_packet_tx}

Suppose that the transmitter is going to send $s$ bits to the receiver. These bits are first encoded in the transmitter and converted into the $m$ symbols. Then, they are transmitted to the receiver as a packet with $m$ symbols. Since the length of the packet is finite, it is received in the receiver with the error probability of $\epsilon$ and the transmission rate is written as \cite{p15}:
\begin{equation}
\frac{s}{m}
\end{equation}
In other words, in $m$ times of using the channel, $s$ bits are transmitted from the transmitter to the receiver. Dissimilar codes have different functions. For a certain amount of error $\epsilon$ in the signal to noise ratio (SNR) $\gamma$, the maximum value of rate is equal to $r$ and defined as \cite{p15,p39m,p40m}:
\begin{equation}\label{eq:r}
r=\log_2(1+\gamma)-\sqrt{\frac{\gamma(\gamma+2)}{m(\gamma+1)^2}}Q^{-1}(\epsilon)\log_2(\e)+\frac{\log_2(m)}{m}
\end{equation}
where $Q^{-1}(.)$ depicts the inverse of Q-function \cite{p15,p39m,p40m}. The throughput rate $r$ in \eqref{eq:r} shows the maximum number of bits that are transmitted from the transmitter to the receiver with the error rate less than or equal to $\epsilon$ in $m$ times of using the normalized time and bandwidth channel. Therefore, the unit of $r$ is expressed in bits per channel use (bits/c.u.). If $m$ tends to the infinity and $\epsilon$ tends to zero, $r$ in \eqref{eq:r} tends to the famous Shannon's capacity \cite{p15,p41m}. Therefore, sometimes $r$ is considered as the channel capacity with a limited number of transmission \cite{p15}. It should be noted that according to \eqref{eq:r}, $r$ is a function of $\gamma$, $m$ and $\epsilon$; however, it is abbreviated as $r$ for brevity.

\subsection{Two-Way Relaying}\label{subsec:two_way_relaying}

Suppose that nodes \A and \B are going to transmit packets to each other and this data transmission is done through the relay \R. There is also no direct path between \A and \B. Normally, in the first time slot, node \A transmits its packet to relay \R. In the second slot, the relay re-sends the packet to node \B. Similarly, in the next two time slots, i.e. the third and fourth slots, node \B sends a packet to node \A. Therefore, four time slots are required to share two packets between two nodes. This transmission method is not cost-effective in terms of spectral efficiency and transmission delay. However, in the HD and FD two-way relays (Figures \ref{fig:1a} and \ref{fig:1b}), just two and one time slots are required to transmit a packet from node \A to \B and vice versa, respectively.

\begin{figure}
	\centering
	\subfloat[Half-Duplex (HD)]{\includegraphics[width=0.85\linewidth]{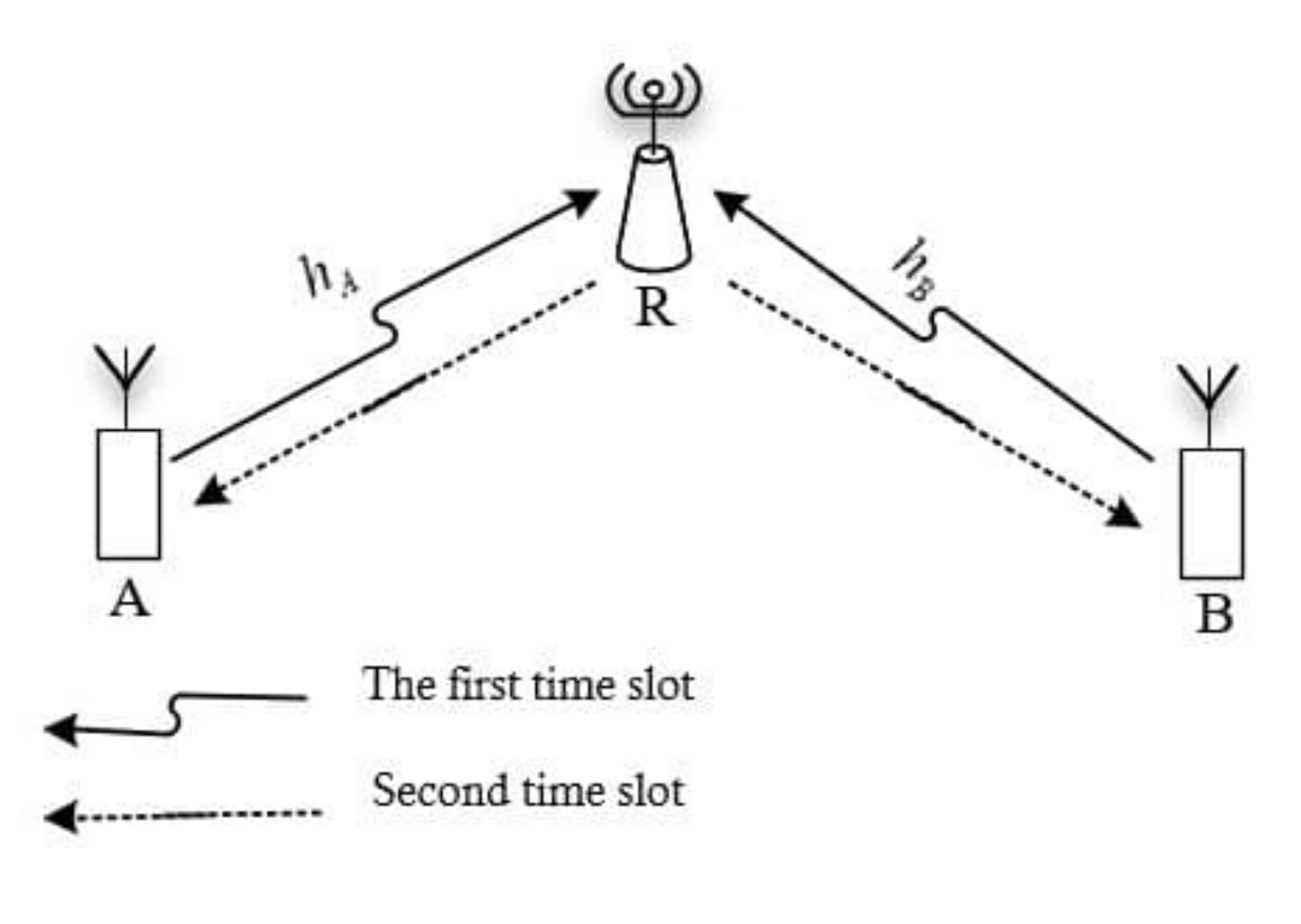}
		\label{fig:1a}}
	\hfil
	\subfloat[Full-Duplex (FD)]{\includegraphics[width=0.85\linewidth]{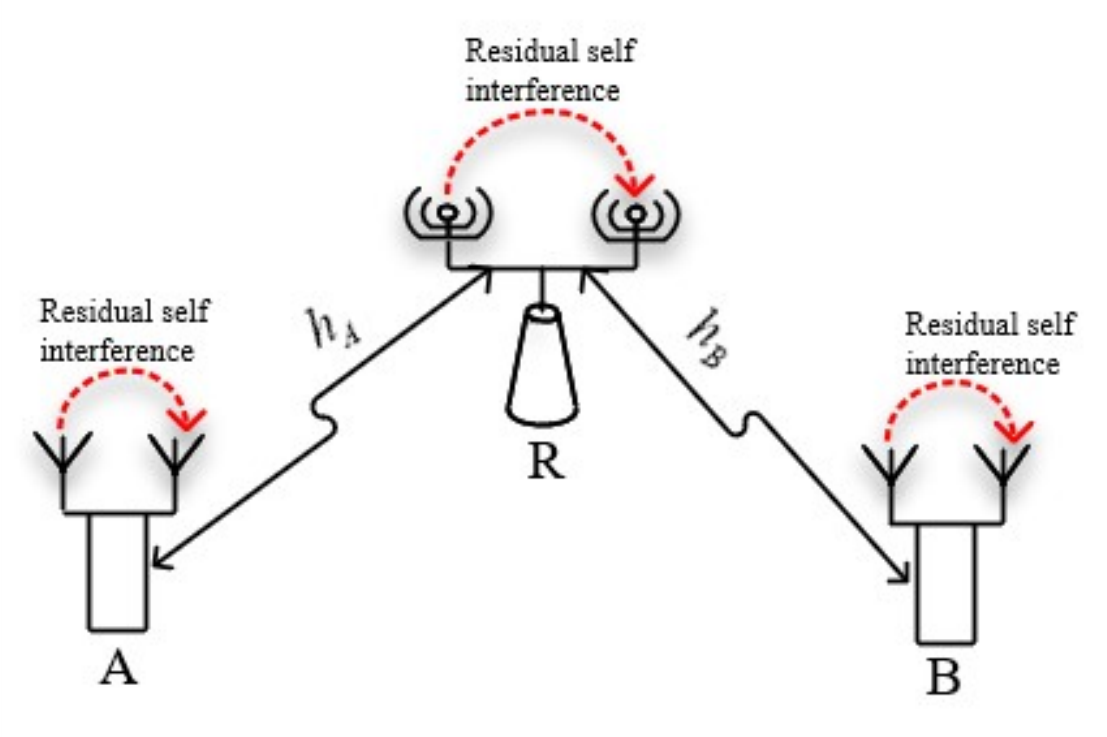}
		\label{fig:1b}}
	\hfil	
	\caption{Two-way relaying.}
	\label{fig:1}
\end{figure}

Suppose that node \A is going to send the short packet $x_{\A}$  with the length $m_{\A}$ to node \B and node \B is going to send the short packet $x_{\B}$ with the length $m_{\B}$ to node \A. Both nodes send their packets to the relay \R simultaneously. Node \A sends its packet with power $P_{\A}$ and node \B sends its packet with power $P_{\B}$. For the two-way relay with HD performance, after normalizing the power, the relay re-transmits the received packets for both nodes with power $P_{\R}$ in the next slot. Both nodes know their transmitted packets. Therefore, if they have access to the channel coefficients, they could remove the effect of their own packet from the received packet from the relay completely. In this case

\begin{equation}\label{eq:HD_y_A}
y_{\A}=\sqrt{P_{\R}}h_{\A}\beta\left(\sqrt{P_{\B}}h_{\B}x_{\B}+w_{\R}\right)+w_{\A}
\end{equation}
specifies the received packet in node \A. In \eqref{eq:HD_y_A}, $h_{\A}$ and $h_{\B}$ represent the coefficient of the flat fading channel between node \A and relay \R and node \B and relay \R. Moreover, $\beta$ is the variable gain of the relay to normalize the power of the received packet, $w_{\A}$ and $w_{\R}$ are the white Gaussian noise with a unique power in the node \A and the relay. The received Packet is written in node \B similarly. By simplifying the relationship and assuming that the powers of $x_{\A}$ and $x_{\B}$ are equal to one, the Signal to Noise Ratio (SNR) in nodes \A and \B is calculated as follows when $H_{\A}=|h_{\A}|^2$ is the power gain of the channel between node \A and the relay and $H_{\B}=|h_{\B}|^2$ is the power gain of the channel between node \B and the relay \cite{p47}.

\begin{equation}\label{eq:HD_gamma_A}
\gamma_{\A}=\frac{P_{\B}P_{\R}H_{\A}H_{\B}}{P_{\R}H_{\A}+P_{\A}H_{\A}+P_{\B}H_{\B}+1}
\end{equation}

\begin{equation}\label{eq:HD_gamma_B}
\gamma_{\B}=\frac{P_{\A}P_{\R}H_{\A}H_{\B}}{P_{\R}H_{\B}+P_{\B}H_{\B}+P_{\A}H_{\A}+1}
\end{equation}

In the two-way relay with FD function, the relay normalizing the power and re-sends the received packets to both nodes with power $P_{\R}$ in the same time. The two nodes receive the transmitted packet at the same time slot too. Both nodes and relays have self-interference in addition to noise because of simultaneous sending and receiving. Again, the nodes know their transmitted packets and they can remove them. Therefore, the received packet in node \A is

\begin{equation}\label{eq:FD_y_A}
y_{\A}=\sqrt{P_{\R}}h_{\A}\beta\left(\sqrt{P_{\B}}h_{\B}x_{\B}+\sqrt{P_{\R}}\omega_{\R}x_{\R}+w_{\R}\right)+\sqrt{P_{\A}}\omega_{\A}x_{\A}+w_{\A}
\end{equation}
The received packet by node \B is re-written in the same manner. The repeated parameters are the same as the HD mode. $x_{\R}$ is the received packet by the relay after normalizing the power or the same as the transmitted packet by the relay. However, $0 \le\omega_{\A}\le 1$ and $0 \le\omega_{\R}\le 1$ indicate the coefficient of self-interference cancellation in node \A and relay \R. In other words, $\sqrt{P_{\A}}\omega_{\A}x_{\A}$ is the residual self-interference in node \A and $\sqrt{P_{\R}}\omega_{\R}x_{\R}$ represents the residual self-interference in relay. If $\omega_{\A}=\omega_{\R}=0$, then the self-interference is removed completely. Otherwise, If $\omega_{\A}=\omega_{\R}=1$, then there is no self-interference cancellation at node \A and relay and the FD relay is likely to work improperly. Therefore, interference management in FD relays is very important. It is worth noting to explain why it is assumed that each node can remove its packet re-sent through the relay but it cannot remove its self-interference completely. The justification is that self-interference enters the receiver with high power. Because of the indeterminate and nonlinear effects of different parts of the transmitter and receiver, it cannot be removed completely from the receiver. In contrast, the re-sent packet of each node to itself enters the receiver with low power. After the interference removal process, the residual part is insignificant and within the noise power limit. Therefore, it is not considered as a separate interference. By simplifying \eqref{eq:FD_y_A}, the Signal to Interference plus Noise Ratio (SINR) in nodes \A and \B appears as \cite{p47}:

\begin{equation}\label{eq:FD_gamma_A}
\gamma_{\A}=\frac{P_{\B}P_{\R}H_{\A}H_{\B}}{P_{\R}^2H_{\A}\Omega_{\R}+P_{\R}H_{\A}+(P_{\A}\Omega_{\A}+1)(P_{\A}H_{\A}+P_{\B}H_{\B}+P_{\R}\Omega_{\R}+1)}
\end{equation}

\begin{equation}\label{eq:FD_gamma_B}
\gamma_{\B}=\frac{P_{\A}P_{\R}H_{\A}H_{\B}}{P_{\R}^2H_{\B}\Omega_{\R}+P_{\R}H_{\B}+(P_{\B}\Omega_{\B}+1)(P_{\A}H_{\A}+P_{\B}H_{\B}+P_{\R}\Omega_{\R}+1)}
\end{equation}
Since the residual interference cannot be estimated, \eqref{eq:FD_gamma_A} and \eqref{eq:FD_gamma_B} use the mean residual power interference coefficient as $\Omega_{\R}=\mathbb{E}\{\omega_{\R}^2\}$, $\Omega_{\A}=\mathbb{E}\{\omega_{\A}^2\}$, and  $\Omega_{\B}=\mathbb{E}\{\omega_{\B}^2\}$. Then, the following assumptions are made to simplify the problem without losing its general concept.
$$m_{\A}=m_{\B}=m$$
$$P_{\A}=P_{\B}=P$$
$$\Omega_{\A}=\Omega_{\B}=\Omega_{\R}=\Omega$$

\subsection{Effective Capacity}\label{subsec:effective_capacity}
The effective capacity is defined as the input rate of the packets to a system's buffer, which can pass through the channel and a certain quality of service could be guaranteed for it \cite{p42m,p43m}. Quality of service refers to the statistical delay in the transmitter buffer, which is defined as \cite{p42m,p43m,p44m}:
\begin{equation}
\mathbb{P}\{D>D_0\}\approx\eta\e^{-\theta D_0}
\end{equation}
where $\mathbb{P}\{.\}$ is the probability of occurrence, $D$ is the packet delay in the transmitting buffer, $D_0$ is a threshold for the delay, $\eta$ is the probability that the transmitter buffer is not empty and $\theta$ denotes the quality of service exponent. As the value of $\theta$ is set to a higher level, the probability of the delay is reduced and the quality of service will be stricter. In contrast, for a lower value of  $\theta$, the probability of delay is higher and the quality of service is not too strict. After specifying the amount of requested quality by each user using the parameter $\theta$, the effective capacity $R_E$ is calculated as:
\begin{equation}\label{eq:general_effective_capacity}
R_E=-\frac{1}{m\theta}\log\left[\mathbb{E}\{\e^{-rm\theta}(1-\epsilon)+\epsilon\}\right].
\end{equation}
It should be noted that $\log(.)$ represents a natural logarithm function.

\subsection{Communication Channel}\label{subsec:comm_channel}
The communication channel between the nodes and the relay is considered  as Rayleigh flat fading. The channel coefficient between nodes \A and \B with relay \R is equal to $h_{\A}=h_{0,\A}d_{\A}^{-\alpha/2}$ and $h_{\B}=h_{0,\B}d_{\B}^{-\alpha/2}$ respectively. $h_{0,\A}$ and $h_{0,\B}$ represent uncorrelated complex Gaussian coefficients with zero mean and unit variance. Moreover, $d_{\A}$ and $d_{\B}$ determine the distance between nodes \A and \B to the relay and $d_{\A}+d_{\B}=1$. $\alpha$ denotes the path loss exponent of the channel. 

\subsection{Table of Parameters}\label{subsec:table_params}
For convenience, the main variables of the article are presented in Table \ref{tab:1}. To summarize some of the relations and also in Table \ref{tab:1}, the index $i$ displays one of the nodes \A or \B $(i\in\{\A,\B\})$. Moreover, $\bar{i}$ depicts the complementary of $i$ in the set. So, if $i=\A$, then $\bar{i}=\B$.

\begin{table}
  \centering
  \caption{Main parameters of this paper.}\label{tab:1}
  \begin{tabular}{|c|c|}
  \hline
  \textbf{Abbreviation} & \textbf{Description} \\
  \hline\hline
	$m$ & Short packet length \\
  \hline
    $\epsilon$ & Error probability of received short packet \\
  \hline
    $P$ & Transmitted power of node \A and \B \\
  \hline
    $P_{\R}$ & Transmitted power of relay \R \\
  \hline
	$P_{tot}$ & Total available power which is divided between nodes and relay \\
  \hline
	$\Omega$ & Mean residual power interference coefficient \\
  \hline
	$r_i$ & Throughput rate of node $i$ \\
  \hline 
	$H_i$ &  Power gain of the channel between node $i$ and relay \R \\
  \hline 
	$\gamma_i$ & SNR of node $i$ in HD mode or SINR of node $i$ in FD mode \\
  \hline
	$\theta_i$ & Quality of service exponent in node $i$ \\
  \hline
	$R_{E,i}$ & Effective capacity of node $i$ \\
  \hline  
\end{tabular}
\end{table}


\section{Maximizing the Effective Capacity of the Two-Way Relay}\label{sec:maximizing_effective_capacity}

In this section, the appropriate power allocation between nodes \A and \B and relay \R is calculated. The allocated power maximize the effective capacity when short packet is transmitted by node \A and node \B. It is assumed that the total available power i.e. $P_{tot}$ is optimally distributed between the nodes and relay. This type of power allocation is commonly used in FD amplify and forward relay \cite{p48}, in FD wireless power communication networks \cite{p49} and two-way fixed-gain amplify and forward relay \cite{p50}. Therefore, $P_{\R}+2P=P_{tot}$. Here, the problem is provided and solved for the HD and FD two-way relays individually. 

The condition of two nodes \A and \B are different and their requested rate and quality of service is not necessarily the same. Therefore, the optimization problem is considered and solved as a multi-objective problem. In the multi-objective optimization problem, there are two or more objective functions to be optimized instead of one single objective function. These functions are usually not aligned and do not optimize at the same point. Therefore, the multi-objective optimization problems have a set of optimal solutions called Pareto optimal solutions instead of a specific optimal solution \cite{p26}. In multi-objective minimization (multi-objective maximization), the optimal Pareto point is selected when there is no other point in the feasible design space in which the value of all objective functions is less than or equal to (greater than or equal to) the value of objective functions at the optimal Pareto point \cite{p26}. Finally, the expert user suggests a point as an answer according to the priority of the objective functions if it is necessary. Solving the multi-objective optimization problem has various techniques. For more details see \cite{p26}.

Before designing and solving the problem, lemma \ref{lem:1} is addressed. This theorem will be necessary for solving the discussed problems. 
\begin{lemma}\label{lem:1}
If $\gamma>1$, $m>100$ and $\epsilon>10^{-23}$, the throughput rate in \eqref{eq:r} is strictly increasing and concave in terms of $\gamma$.
$$\forall\{\gamma>1,~m>100,~\epsilon>10^{-23}\}\implies\left\{\frac{\partial r}{\partial\gamma}>0,~\frac{\partial^2r}{\partial\gamma^2}\le0\right\}$$
\end{lemma}

\begin{proof}
Please refer to \cite{p37m}.
\end{proof}

In most communication systems, a high SNR (about 10 and above) is desired under normal constraints. Therefore, in most cases we have $\gamma>1$. The constraint $m>100$ is not restrictive at all. That is because $100$ bits equals $12.5$ bytes, which is not a significant value in the transmission of the short packets. The error rate of $\epsilon=10^{-23}$ is insignificant and it is never required. Thus, the constraint of lemma \ref{lem:1} is not strict, and lemma \ref{lem:1} is always approved. 

\subsection{Multi-Objective Optimization Problem in the Two-Way HD Relay}\label{subsec:multi_objective_HD}

The multi-objective optimization problem in a two-way HD relay with short packet transmission to maximize the effective capacity of two nodes is as follows:
\begin{equation}\label{eq:HD_multi_objective_problem_1}
\left\{ \begin{array}{l}
\mathop {\max }\limits_{{P_{\R}}} {\rm{   }}{R_{E,{\A}}}\\
\mathop {\max }\limits_{{P_{\R}}} {\rm{   }}{R_{E,{\B}}}\\
{\rm{s}}{\rm{.t}}{\rm{.   }}{P_{\R}} < {P_{tot}}
\end{array} \right.
\end{equation}
In the optimization problem \eqref{eq:HD_multi_objective_problem_1}, the effective capacity of the nodes is written according to \eqref{eq:general_effective_capacity} as follows:
\begin{equation}\label{eq:HD_effective_capacity}
R_{E,i}=-\frac{1}{m\theta_i}\log\left[\mathbb{E}\{\e^{-r_i\frac{m}{2}\theta_i}(1-\epsilon_i)+\epsilon_i\}\right],~~~i\in\{\A,\B\}.
\end{equation}
We use $\frac{m}{2}$ in \eqref{eq:HD_effective_capacity} because in the HD relay, the number of using channel $m$ is twice the length of the transmitted packet. Now, it is proved that the optimization problem \eqref{eq:HD_multi_objective_problem_1} is concave. Accordingly, first lemma \ref{lem:2} and \ref{lem:3} are proved, and finally by lemma \ref{lem:4}, the concavity of the problem \eqref{eq:HD_multi_objective_problem_1} is presented. 
\begin{lemma}\label{lem:2}
The SNR is concave in terms of the relay power $P_{\R}$ and only has a single maximum point. That is:
$$\frac{\partial^2\gamma_i}{\partial P_{\R}^2}<0,~~~i\in\{\A,\B\}$$
\end{lemma}

\begin{proof}
For ease of writing, the lemma is proved for node \A. The relations can also be simply written and proved for node \B. Using \eqref{eq:HD_gamma_A}, the second-order derivative of $\gamma_{\A}$ relative to $P_{\R}$ is:
\begin{equation}\label{eq:13}
\frac{{{\partial ^2}{\gamma _{\A}}}}{{\partial {P_{\R}}^2}} =  - \frac{{4{H_{\A}}{H_{\B}}\left( {{H_{\A}}{P_{tot}} + 1} \right)\left( {{H_{\A}}{P_{tot}} + {H_{\B}}{P_{tot}} + 2} \right)}}{{{{\left( {{H_{\A}}({P_{tot}} + {P_{\R}}) + {H_{\B}}({P_{tot}} - {P_{\R}}) + 2} \right)}^3}}}
\end{equation}
Since $P_{\R}<P_{tot}$, the denominator of \eqref{eq:13} is always positive. Thus, the second derivative in \eqref{eq:13} is constantly negative and $\gamma_{\A}$ is concave relative to $P_{\R}$. Now, the first derivative is as follows:
\begin{equation}\label{eq:14}
\frac{{\partial {\gamma _{\A}}}}{{\partial {P_{\R}}}} = \frac{{{H_{\A}}{H_{\B}}({P_{tot}} - 2{P_{\R}})}}{{{H_{\A}}({P_{tot}} + {P_{\R}}) + {H_{\B}}({P_{tot}} - {P_{\R}}) + 2}}
- \frac{{{H_{\A}}{H_{\B}}{P_{\R}}({H_{\A}} - {H_{\B}})({P_{tot}} - {P_{\R}})}}{{{{\left( {{H_{\A}}({P_{tot}} + {P_{\R}}) + {H_{\B}}({P_{tot}} - {P_{\R}}) + 2} \right)}^2}}}
\end{equation}
If \eqref{eq:14} is set equal to zero, a quadratic equation is obtained with two roots. These two roots must be real. By calculating the roots and the second-order derivative in both roots, one of the roots will be acceptable because the second-order derivative becomes positive in the other root, which is contrary to the second-order derivative under relation \eqref{eq:13}. $P_{\R}^{\star\A}$ in \eqref{eq:15} is the optimal power of the relay, which makes the first-order derivative zero and maximizes $\gamma_{\A}$.
\begin{equation}\label{eq:15}
P_{\R}^{ \star \A} = 
 - \frac{{{H_{\A}}{P_{tot}} + {H_{\B}}{P_{tot}} + 2 - \sqrt {2({H_{\A}}{P_{tot}} + 1)({H_{\A}}{P_{tot}} + {H_{\B}}{P_{tot}} + 2)} }}{{{H_{\A}} - {H_{\B}}}}
\end{equation}
\end{proof}

\begin{lemma}\label{lem:3}
$R_{E,\A}$ relative to $r_{\A}$ and $R_{E,\B}$ relative to $r_{\B}$ is strictly increasing and concave. This means that:
$$\frac{\partial R_{E,i}}{\partial r_i}>0,~~~\frac{\partial^2R_{E,i}}{\partial r_i^2}<0,~~~i\in\{\A,\B\}.$$
\end{lemma}

\begin{proof}
The first and second-order derivatives are calculated according to \eqref{eq:16} and \eqref{eq:17}. In \eqref{eq:16}, since $\epsilon_i\le0.5$  and the exponential function is always positive, the numerator and denominator of the relation are positive. Therefore, the first-order derivative is positive and the effective capacity is strictly increasing relative to the throughput rate. Therefore, the second-order derivative is also negative and the effective capacity is concave relative to the throughput rate.
\begin{equation}\label{eq:16}
\frac{{\partial {R_{E,i}}}}{{\partial {r_i}}} = \frac{1}{2}\frac{{\mathbb{E}\left\{ {{{\mathop{\rm e}\nolimits} ^{ - {r_i}\frac{m}{2}{{{\theta_i}}}}}\left( {1 - {\epsilon _i}} \right)} \right\}}}{{\mathbb{E}\left\{ {{{\mathop{\rm e}\nolimits} ^{ - {r_i}\frac{m}{2}{{\theta_i}}}}\left( {1 - {\epsilon _i}} \right) + {\epsilon _i}} \right\}}},{\rm{ }}i \in \{ A,B\}
\end{equation}

\begin{equation}\label{eq:17}
\begin{array}{l}
\frac{{{\partial ^2}{R_{E,i}}}}{{\partial r_i^2}} = \\
{\rm{  }}\frac{1}{2}\frac{{\frac{{ - {\theta _i}m}}{2}\mathbb{E}\left\{ {{{\mathop{\rm e}\nolimits} ^{ - {r_i}\frac{m}{2}{{{\theta_i}}}}}\left( {1 - {\epsilon _i}} \right)} \right\}\left( {\mathbb{E}\left\{ {{{\mathop{\rm e}\nolimits} ^{ - {r_i}\frac{m}{2}{{{\theta_i}}}}}\left( {1 - {\epsilon _i}} \right)} \right\} + {\epsilon _i}} \right)}}{{{{\left( {\mathbb{E}\left\{ {{{\mathop{\rm e}\nolimits} ^{ - {r_i}\frac{m}{2}{{{\theta_i}}}}}\left( {1 - {\epsilon _i}} \right) + {\epsilon _i}} \right\}} \right)}^2}}} + 
{\rm{  }}\frac{1}{2}\frac{{\frac{{{\theta _i}m}}{2}{{\left( {\mathbb{E}\left\{ {{{\mathop{\rm e}\nolimits} ^{ - {r_i}\frac{m}{2}{{{\theta_i}}}}}\left( {1 - {\epsilon _i}} \right)} \right\}} \right)}^2}}}{{{{\left( {\mathbb{E}\left\{ {{{\mathop{\rm e}\nolimits} ^{ - {r_i}\frac{m}{2}{{{\theta_i}}}}}\left( {1 - {\epsilon _i}} \right) + {\epsilon _i}} \right\}} \right)}^2}}} = \\
{\rm{  }}\frac{{ - {\theta _i}m{\epsilon _i}}}{4}\frac{{\mathbb{E}\left\{ {{{\mathop{\rm e}\nolimits} ^{ - {r_i}\frac{m}{2}{{{\theta_i}}}}}\left( {1 - {\epsilon _i}} \right)} \right\}}}{{{{\left( {\mathbb{E}\left\{ {{{\mathop{\rm e}\nolimits} ^{ - {r_i}\frac{m}{2}{{{\theta_i}}}}}\left( {1 - {\epsilon _i}} \right) + {\epsilon _i}} \right\}} \right)}^2}}},{\rm{ }}i \in \{ \A,\B\} 
\end{array}
\end{equation}

\end{proof}

\begin{lemma}\label{lem:4}
If lemma \ref{lem:1} holds, the effective capacity is concave relative to the relay power $P_{\R}$ and only has one maximum point. This means that:
$$\frac{\partial^2R_{E,i}}{\partial P_{\R}^2}<0,~~~i\in\{\A,\B\}$$.
\end{lemma}

\begin{proof}
This lemma is again proved for node \A. Proof for node \B is done in the same manner. According to the chain rule, the second-order derivative of $R_{E,\A}$ relative to $P_{\R}$ equals
\begin{equation}\label{eq:18}
\frac{\partial^2 R_{E,\A}}{\partial P_{\R}^2}=\frac{\partial^2 R_{E,\A}}{\partial r_{\A}^2}\left(\frac{\partial r_{\A}}{\partial P_{\R}}\right)^2+\frac{\partial^2 r_{\A}}{\partial P_{\R}^2}+\frac{\partial R_{E,\A}}{\partial r_{\A}}.
\end{equation}
On the right side of \eqref{eq:18}, according to lemma \ref{lem:3}, the first fraction is negative. The second fraction is always positive. The fourth fraction is also positive according to lemma \ref{lem:3}. If it is proved that the third fraction is constantly negative, the negation of \eqref{eq:18} will be proved. Therefore, using the chain rule, the third fraction is rewritten as:
\begin{equation}\label{eq:19}
\frac{\partial^2 r_{\A}}{\partial P_{\R}^2}=\frac{\partial^2 r_{\A}}{\partial \gamma_{\A}^2}\left(\frac{\partial \gamma_{\A}}{\partial P_{\R}}\right)^2+\frac{\partial^2 \gamma_{\A}}{\partial P_{\R}^2}+\frac{\partial r_{\A}}{\partial \gamma_{\A}}.
\end{equation}
On the right side of \eqref{eq:19}, according to lemma \ref{lem:1}, the first fraction is negative and the fourth fraction is positive. According to lemma \ref{lem:2}, the second fraction is also positive and the third fraction is negative. Therefore, \eqref{eq:19} is negative and it is possible to claim that the effective capacity of node \A is concave relative to the power of the relay. Now, the first-order derivative is written as follows:
\begin{equation}\label{eq:20}
\frac{\partial R_{E,\A}}{\partial P_{\R}}=\frac{\partial R_{E,\A}}{\partial r_{\A}}\frac{\partial r_{\A}}{\partial \gamma_{\A}}\frac{\partial \gamma_{\A}}{\partial P_{\R}}.
\end{equation}
On the right side of \eqref{eq:20}, according to lemma \ref{lem:3}, the first fraction is non-zero. According to lemma \ref{lem:2}, the second fraction is not zero. Based on lemma \ref{lem:2}, the third fraction becomes zero only at the point $P_{\R}^{\star\A}$. Therefore, it can be concluded that the effective capacity of node \A is concave relative to $P_{\R}$ and reaches its maximum value only at the point $P_{\R}^{\star\A}$.
\end{proof}

Since the functions of the multi-objective problem \eqref{eq:HD_multi_objective_problem_1} are concave, the weighted sum is a proper solution for them. The weighted sum method is very simple and it is the necessary and sufficient condition to extract all the optimal Pareto points in the case of multi-objective concave (or convex) problems \cite{p26}. In the weighted sum method, the first objective function of the optimization problem \eqref{eq:HD_multi_objective_problem_1} with the coefficient $w$ and the second objective function with the coefficient $1-w$ are summed up and they form a new objective function. The coefficient $0\le w\le1$ presents the priority of the first objective function and the coefficient $0\le 1-w\le1$ indicates the priority value of the second objective function. Accordingly, the optimization problem \eqref{eq:HD_multi_objective_problem_1} is simplified to
\begin{equation}\label{eq:HD_multi_objective_problem_2}
\begin{array}{l}
\left\{ \begin{array}{l}
\mathop {\min }\limits_{{P_{\R}}} {\rm{   }}J\\
{\rm{s}}{\rm{.t}}{\rm{.   }}{P_{\R}} < {P_{tot}}
\end{array} \right.,\\
J = \frac{w}{{m{\theta _A}}}\log \left[ {\mathbb{E}\left\{ {{{\mathop{\rm e}\nolimits} ^{ - {r_{\A}}\frac{m}{2}{{{\theta_{\A}}}}}}\left( {1 - {\epsilon _{\A}}} \right) + {\epsilon _{\A}}} \right\}} \right]\\
{\rm{  }} + \frac{{1 - w}}{{m{\theta _{\B}}}}\log \left[ {\mathbb{E}\left\{ {{{\mathop{\rm e}\nolimits} ^{ - {r _{\B}}\frac{m}{2}{{{\theta_{\B}}}}}}\left( {1 - {\epsilon _{\B}}} \right) + {\epsilon _{\B}}} \right\}} \right].
\end{array}
\end{equation}
It should be noted that the first and second objective functions of the problem \eqref{eq:HD_multi_objective_problem_1} are summed up with the negative sign. Thus, maximization of the problem \eqref{eq:HD_multi_objective_problem_1} is minimized in problem \eqref{eq:HD_multi_objective_problem_2}.

Due to the statistical expectation in the definition of effective capacity, the numerical solution of problem \eqref{eq:HD_multi_objective_problem_2} is time-consuming. Therefore, the problem is simplified to achieve a faster solution. To simplify $J$ in \eqref{eq:HD_multi_objective_problem_2}, due to the small probability of error, $\epsilon_{\A}$ and $\epsilon_{\B}$ are ignored and $J$ is simplified to
\begin{equation}\label{eq:HD_J_1}
J \approx \frac{w}{{m{\theta _{\A}}}}\log \left[ {\mathbb{E}\left\{ {{{\mathop{\rm e}\nolimits} ^{ - {r _{\A}}\frac{m}{2}{{{\theta_{\A}}}}}}\left( {1 - {\epsilon _{\A}}} \right)} \right\}} \right]
+ \frac{{1 - w}}{{m{\theta _{\B}}}}\log \left[ {\mathbb{E}\left\{ {{{\mathop{\rm e}\nolimits} ^{ - {r _{\B}}\frac{m}{2}{{{\theta_{\B}}}}}}\left( {1 - {\epsilon _{\B}}} \right)} \right\}} \right].
\end{equation}
Then, using the approximation of logarithm of the sum of the exponential functions according to (23) \cite{p46m},
\begin{equation}\label{eq:log_sum_approx}
\log\mathbb{E}\{\e^z\}\approx \max\{z\}+C
\end{equation}
$J$ is again simplified to
\begin{equation}\label{eq:HD_J_2}
\begin{array}{l}
J \approx \frac{w}{{m{\theta _{\A}}}}\left[ {\max\left\{ {{{ - {r _{\A}}\frac{m}{2}{{{\theta_{\A}}}}}}+C+\log\left( {1 - {\epsilon _{\A}}} \right)} \right\}} \right]
+ \frac{{1 - w}}{{m{\theta _{\B}}}} \left[  {\max\left\{ {{{ - {r _{\B}}\frac{m}{2}{{{\theta_{\B}}}}}}+C+\log\left( {1 - {\epsilon _{\B}}} \right)} \right\}} \right]\\
=\left[\max\{-\frac{w}{2}r_{\A}\}+\frac{w}{m\theta_{\A}}C+\frac{w}{m\theta_{\A}}\log(1-\epsilon_{\A})\right]+
\left[\max\{-\frac{1-w}{2}r_{\B}\}+\frac{1-w}{m\theta_{\B}}C+\frac{1-w}{m\theta_{\B}}\log(1-\epsilon_{\B})\right]
\end{array}
\end{equation}
In \eqref{eq:log_sum_approx}, $z\in\mathbb{Z}$ denotes an arbitrary random variable and $\max{z}$ specifies the largest value of a random variable in the set $\mathbb{Z}$. In addition, $C$ is a constant value and does not depend on a random variable  . This constant value also appears in \eqref{eq:HD_J_2}. Now, $J$ is inserted in the optimization problem \eqref{eq:HD_multi_objective_problem_2} and the simplified problem is written as:
\begin{equation}\label{eq:HD_multi_objective_problem_3}
\left\{ \begin{array}{l}
\mathop {\min }\limits_{{P_{\R}}} ~~~\left\{\max\{-\frac{w}{2}r_{\A}\}+\max\{-\frac{1-w}{2}r_{\B}\}+\prime{C}\right\}\\
{\rm{s}}{\rm{.t}}{\rm{.   }}{P_{\R}} < {P_{tot}}
\end{array} \right.
\end{equation}
and $\prime{C}$ indicates fixed parameters that do not affect optimization.

The optimization problem \eqref{eq:HD_multi_objective_problem_3} is a min-max optimization problem \cite{p26}. Therefore, the solution to this type of problem is used here as well. First, two auxiliary variables $\tau_{\A}$ and $\tau_{\B}$ are defined and the problem \eqref{eq:HD_multi_objective_problem_3} is converted to 
\begin{equation}\label{eq:HD_min_max_problem_1}
\left\{ \begin{array}{l}
\mathop {\min }\limits_{{P_{\R}}} {\rm{   }}{\tau _{\A}} + {\tau _{\B}}\\
{\rm{s}}{\rm{.t}}{\rm{.   }}{\tau _{\A}} \ge  - \frac{w}{2}{r_{\A}}\\
{\rm{s}}{\rm{.t}}{\rm{.   }}{\tau _{\B}} \ge  - \frac{{1 - w}}{2}{r_{\A}}\\
{\rm{s}}{\rm{.t}}{\rm{.   }}{P_{\R}} < {P_{tot}}
\end{array} \right.
\end{equation}
For further simplification, using the auxiliary variable $\tau=\tau_{\A}+\tau_{\B}$, the optimization problem is rewritten again as:
\begin{equation}\label{eq:HD_min_max_problem_2}
\left\{ \begin{array}{l}
\mathop {\min }\limits_{{P_{\R}}} {\rm{   }}\tau \\
{\rm{s}}{\rm{.t}}{\rm{.   }}\tau  \ge  - \frac{w}{2}{r_{\A}} - \frac{{1 - w}}{2}{r_{\B}}\\
{\rm{s}}{\rm{.t}}{\rm{.   }}{P_{\R}} < {P_{tot}}
\end{array} \right.
\end{equation}
The optimization problem \eqref{eq:HD_min_max_problem_2} has no statistical expectation. Therefore, is much easier to solve than the optimization problem \eqref{eq:HD_multi_objective_problem_2}. To speed up convergence in the numerical solution, the following relationship can be used as a starting point.
$$P_{\R}=wP_{\R}^{\star\A}+(1-w)P_{\R}^{\star\B}$$
The answer to the problem \eqref{eq:HD_min_max_problem_2} is called $P_{\R}^{\star}$. The SNR of nodes \A and \B in this optimal point are $\gamma_{\A}^{\star}$ and $\gamma_{\B}^{\star}$. 

To establish the assumptions of lemma \ref{lem:1}, $\gamma_{\A}^{\star}>1$ and $\gamma_{\B}^{\star}>1$ are needed. Generally,
we assume $\gamma_{\A}^{\star}>\gamma_{T,\A}$ and $\gamma_{\B}^{\star}>\gamma_{T,\B}$. This comparison may support the assumption of lemma \ref{lem:1}, or satisfy the minimum throughput rate of node \A and node \B. Now, if the constraint $\gamma_{\A}^{\star}>\gamma_{T,\A}$ is not met in node \A, node \A is not transmitted and the relay uses its power $P_{\R}^{\star\B}$ to transmit the data of node \B. Similarly, if the constraint $\gamma_{\B}^{\star}>\gamma_{T,\B}$ is not met, node \B is not transmitted and the power $P_{\R}^{\star\A}$ is allocated to the relay.

Lemma \ref{lem:2} shows that the SNR is concave relative to the relay power and only has one maximum point. Therefore, the equation $\gamma_{\A}=\gamma_{T,\A}$ has exactly two roots $P_{T,\A}^1$ and $P_{T,\A}^2$; and the equation $\gamma_{\B}=\gamma_{T,\B}$ must have two roots equal to $P_{T,\B}^1$ and $P_{T,\B}^2$. Therefore, if $P_{\R}^{\star}<P_{T,\A}^1$ or $P_{\R}^{\star}>P_{T,\A}^2$, then node \A is not transmitted. Similarly, if $P_{\R}^{\star}<P_{T,\B}^1$ or $P_{\R}^{\star}>P_{T,\B}^2$, then node \B will have no transmission. In \eqref{eq:28}, $P_{T,i}^1$ and $P_{T,i}^2$ are calculated for $i\in\{\A,\B\}$. Moreover, the value $\sigma_1$ is given in \eqref{eq:29}.
\begin{equation}\label{eq:28}
\begin{array}{l}
P_{T,i}^1 = \frac{{\left( {{H_{\bar i}} - {H_i}} \right){\gamma _{T,i}} + {H_{\A}}{H_{\B}}{\gamma _{T,i}} - {\sigma _1}}}{{2{H_{\A}}{H_{\B}}}}\\
P_{T,i}^2 = \frac{{\left( {{H_{\bar i}} - {H_i}} \right){\gamma _{T,i}} + {H_{\A}}{H_{\B}}{\gamma _{T,i}} + {\sigma _1}}}{{2{H_{\A}}{H_{\B}}}}
\end{array}
\end{equation}
\begin{equation}\label{eq:29}
\begin{array}{l}
{\sigma _1} = \left[ {H_{\A}^2H_{\B}^2P_{tot}^2 - 6H_i^2{H_{\bar i}}} \right.{P_{tot}}{\gamma _{T,i}} - 2{H_i}{H_{\bar i}}^2{P_{tot}}{\gamma _{T,i}}\\
{\rm{    }} + {\left( {{H_{\A}}{\gamma _{T,i}}} \right)^2} + {\left( {{H_{\B}}{\gamma _{T,i}}} \right)^2}\\
\left. {{\rm{    }} - 2{H_{\A}}{H_{\B}}\gamma _{T,i}^2 - 8{H_{\A}}{H_{\B}}{\gamma _{T,i}}} \right]
\end{array}
\end{equation}

\subsection{Multi-Objective Optimization Problem in the Two-Way FD Relay}\label{subsec:multi_objective_FD}
To maximize the effective capacity of nodes in a cooperative system with a two-way FD relay and short packet transmission, the following optimization problem is used.
\begin{equation}\label{eq:FD_multi_objective_problem_1}
\left\{ \begin{array}{l}
\mathop {\max }\limits_{{P_{\R}}} {\rm{   }}{R_{E,{\A}}}\\
\mathop {\max }\limits_{{P_{\R}}} {\rm{   }}{R_{E,{\B}}}\\
{\rm{s}}{\rm{.t}}{\rm{.   }}{P_{\R}} < {P_{tot}}
\end{array} \right.
\end{equation}
In this problem, the effective capacity of nodes \A and \B can be written as:
\begin{equation}\label{eq:FD_effective_capacity}
{R_{E,i}} =  - \frac{1}{{m{\theta _i}}}\log \left[ {\left\{ {{{\mathop{\rm e}\nolimits} ^{ - {r _i}m{\theta_i}}}\left( {1 - {\epsilon _i}} \right) + {\epsilon _i}} \right\}} \right],{\rm{ }}i \in \{\A,\B\}. 
\end{equation}
In FD relay, the SINR is not concave relative to the relay power. This means that the second-order derivative of $\gamma_i,~i\in\{\A,\B\}$ is not always negative relative to $P_{\R}$. Moreover, the effective capacity is not concave relative to the relay power and the second-order derivative of $R_{E,i},~i\in\{\A,\B\}$ is not constantly negative relative to $P_{\R}$. Therefore, lemma \ref{lem:2} and \ref{lem:4} are not generally true for the FD relay. However, it is possible to prove that the SINR and effective capacity only have one maximum point relative to $P_{\R}$. Therefore, lemma \ref{lem:5} is presented as follows.
\begin{lemma}\label{lem:5}
The SINR and the effective capacity have only one maximum point relative to the relay power. 
\end{lemma}

\begin{proof}
For simplicity, the proof is provided for node \A. Again, the relations are simply written and proved for node \B. $\gamma_{\A}$ for FD relay is written in \eqref{eq:FD_gamma_A}. Since the value of $\gamma_{\A}$ is constantly positive and becomes zero for $P_{\R}=0$ and $P_{\R}=P_{tot}$, it has a definite maximum point. Then, the quadratic equation is obtained by calculating the first-order derivative relative to $P_{\R}$ and the two values of $P_{\R}$ are calculated as follows:
\begin{equation}\label{eq:32}
\begin{array}{l}
{P_{{\R}}^1},{\rm{ }}{P_{{\R}}^2} = \\
 - \frac{{\left( {{H_{\A}}{P_{tot}} + {H_{\B}}{P_{tot}} + 2} \right)\left( {\Omega {P_{tot}} + 2} \right)}}{{\left( {{H_{\A}} - {H_{\B}}} \right)\left( {\Omega {P_{tot}} + 2} \right) + 2\Omega \left( {{H_{\A}}{P_{tot}} + 1} \right)}}\\
 \pm \frac{{2\sqrt {\left( {{H_{\A}}{P_{tot}} + 1} \right)\left( {\Omega {P_{tot}} + 1} \right)\left( {\Omega {P_{tot}} + 2} \right)\left( {{H_{\A}}{P_{tot}} + {H_{\B}}{P_{tot}} + 2} \right)} }}{{\left( {{H_{\A}} - {H_{\B}}} \right)\left( {\Omega {P_{tot}} + 2} \right) + 2\Omega \left( {{H_{\A}}{P_{tot}} + 1} \right)}}
\end{array}
\end{equation}
It should be noted that the second fraction sign is positive for $P_{\R}^1$ and negative for $P_{\R}^2$. First, it is assumed that $P_{\R}^1$ is the answer and $P_{\R}^1\ge0$. After multiplying both sides of the equation by the fraction's denominators and achieving an algebraic expression, its sign is determined. After simplification, it is proved that $H_{\A}\le0$ which is not true. Therefore, there is only one maximum point as:
\begin{equation}\label{eq:33}
\begin{array}{l}
{P_{{\R}}^{\star\A}}= \\
 - \frac{{\left( {{H_{\A}}{P_{tot}} + {H_{\B}}{P_{tot}} + 2} \right)\left( {\Omega {P_{tot}} + 2} \right)}}{{\left( {{H_{\A}} - {H_{\B}}} \right)\left( {\Omega {P_{tot}} + 2} \right) + 2\Omega \left( {{H_{\A}}{P_{tot}} + 1} \right)}}\\
 - \frac{{2\sqrt {\left( {{H_{\A}}{P_{tot}} + 1} \right)\left( {\Omega {P_{tot}} + 1} \right)\left( {\Omega {P_{tot}} + 2} \right)\left( {{H_{\A}}{P_{tot}} + {H_{\B}}{P_{tot}} + 2} \right)} }}{{\left( {{H_{\A}} - {H_{\B}}} \right)\left( {\Omega {P_{tot}} + 2} \right) + 2\Omega \left( {{H_{\A}}{P_{tot}} + 1} \right)}}.
\end{array}
\end{equation}
Then, similar to the relations of lemma \ref{lem:4} and considering the assumptions of lemma \ref{lem:1}, it is proved that $R_{E,\A}$ is maximized at the point $P_{\R}^{\star\A}$.  
\end{proof}

Although the effective capacity is not concave relative to the relay power, lemma \ref{lem:5} shows that the effective capacity only has one maximum point. Thus, during the numerical solution of the problem, if the solution algorithm converges, the obtained solution will be the global maximum answer. On the other hand, when the self-interference is eliminated properly and $\Omega\to 0$, lemmas \ref{lem:2} and \ref{lem:4} are established and the effective capacity is concave. Therefore, the optimal answer for nodes \A and \B will be close to $P_{\R}^{\star\A}$ and $P_{\R}^{\star\B}$ even if $\Omega$ is small but not necessarily close to zero. Therefore, the following point is applied as a starting point to increase the computing speed in multi-objective optimization.
$$P_{\R}=wP_{\R}^{\star\A}+(1-w)P_{\R}^{\star\B}$$

For non-concave optimization \eqref{eq:FD_multi_objective_problem_1}, the weighted sum method is a sufficient condition to achieve the optimal Pareto point. This means that if this method is applied, the answers will be the optimal Pareto points. This method is not a necessary condition to calculate the optimum points. However, since the weighted sum method is simple, it is applied in FD problem too. If other methods that are the necessary and sufficient condition to achieve the optimal Pareto point are used, the results of the HD problem will no longer comparable with FD. Since the results obtained in HD and FD mode are going to be compared, the weighted sum method will be applied to solve the problem \eqref{eq:FD_multi_objective_problem_1} (articles such as \cite{p21} have used this method to solve the non-convex problem). Therefore, by combining the weighted sum of two objective functions, the multi-objective optimization problem is simplified to
\begin{equation}\label{eq:FD_multi_objective_problem_2}
\begin{array}{l}
\left\{ \begin{array}{l}
\mathop {\min }\limits_{{P_{\R}}} {\rm{   }}J\\
{\rm{s}}{\rm{.t}}{\rm{.   }}{P_{\R}} < {P_{tot}}
\end{array} \right.,\\
J = \frac{w}{{m{\theta _{\A}}}}\log \left[ {\mathbb{E}\left\{ {{{\mathop{\rm e}\nolimits} ^{ - {r_{\A}}m{\theta_{\A}}}}\left( {1 - {\epsilon _{\A}}} \right) + {\epsilon _{\A}}} \right\}} \right]\\
{\rm{  }} + \frac{{1 - w}}{{m{\theta _{\B}}}}\log \left[ {\mathbb{E}\left\{ {{{\mathop{\rm e}\nolimits} ^{ - {r_{\B}}m{\theta_{\B}}}}\left( {1 - {\epsilon _{\B}}} \right) + {\epsilon _{\B}}} \right\}} \right].
\end{array}
\end{equation}
Similar to equations \eqref{eq:HD_J_1} to \eqref{eq:HD_min_max_problem_1}, problem \eqref{eq:FD_multi_objective_problem_2} is also simplified as a min-max problem \eqref{eq:FD_min_max_problem_1}.
\begin{equation}\label{eq:FD_min_max_problem_1}
\left\{ \begin{array}{l}
\mathop {\min }\limits_{{P_{\R}}} {\rm{   }}\tau \\
{\rm{s}}{\rm{.t}}{\rm{.   }}\tau  \ge  - \frac{w}{2}{r_{\A}} - \frac{{1 - w}}{2}{r_{\B}}\\
{\rm{s}}{\rm{.t}}{\rm{.   }}{P_{\R}} < {P_{tot}}
\end{array} \right.
\end{equation}
Once again $w$ and $1-w$ indicate the importance of the effective capacity of nodes \A and \B, respectively, and $r_i,~i\in\{\A,\B\}$ is the throughput rate of node $i$ in the FD relay.


\section{Simulation, Numerical Results, and Comparison}\label{sec:simulation}

This section provides simulations, numerical results, and a comparison of performance in different modes. Here the main parameters include $m=100$, $P_{tot}=1000~W$, $\alpha=4$, $\epsilon_{\A}=\epsilon_{\B}=10^{-4}$, $\theta_{\A}=\theta_{\B}=10^{-3}$, $d_{\A}=0.5$, $d_{\A}+d_{\B}=1$ and $\gamma_{R,\A}=\gamma_{T,B}=1$. If any of these parameters is altered, the new value will be mentioned. 

\begin{figure}
	\begin{center}
		\includegraphics[draft=false,width=0.95\linewidth]{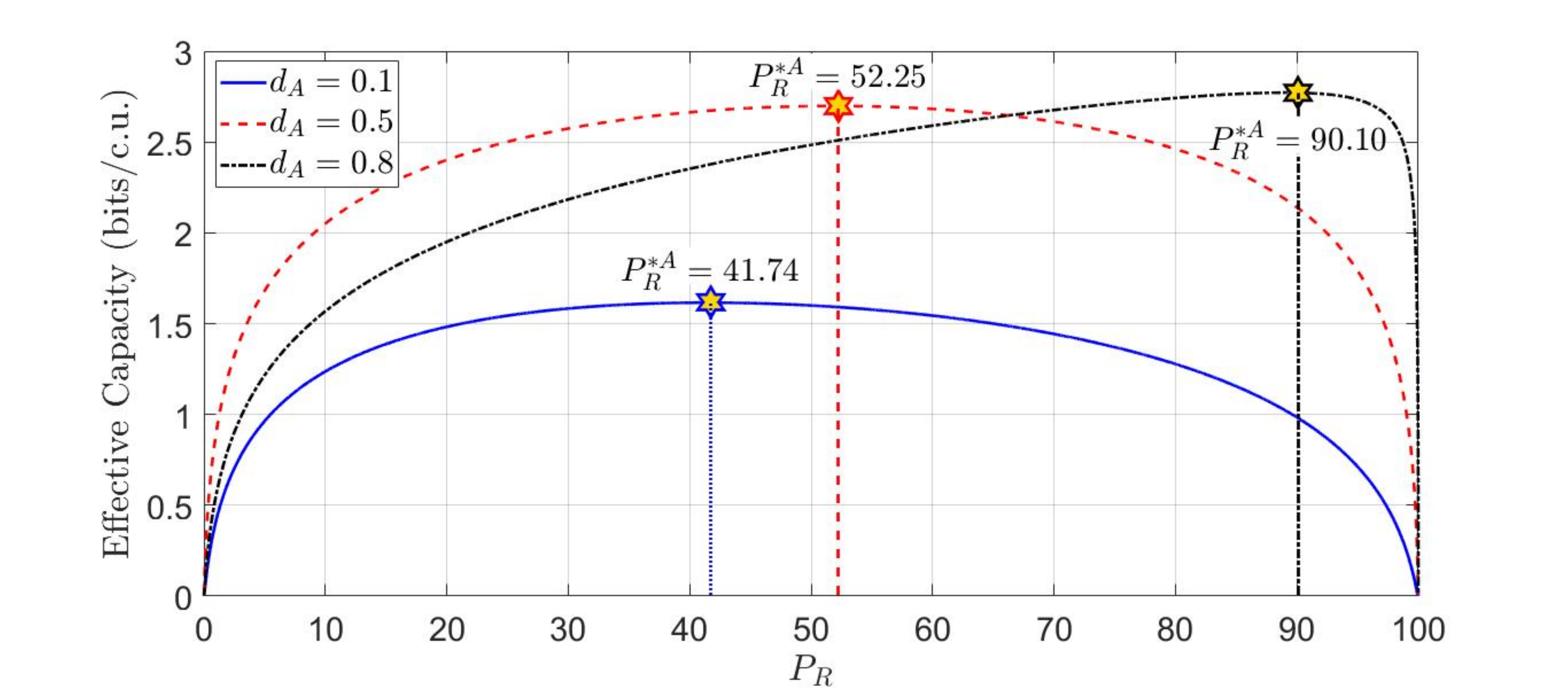}
	\end{center}
	\caption{Effective capacity of node \A vs. HD relay power and different node-relay distance.}
	\label{fig:2}
\end{figure}

Figure \ref{fig:2} shows the effective capacity of node \A in terms of HD relay power for three values of $d_{\A}=0.1$, $d_{\A}=0.5$ and $d_{\A}=0.8$. This figure specifies $P_{\R}^{\star\A}$ for which $R_{E,\A}$ is maximized. According to lemma \ref{lem:4}, the effective capacity of each node is only maximized for a single value of relay power. Meanwhile, if $d_{\A}$ increases, the relay will approach node \B and the packet of node \B will reach the relay even with the small amount of power. On the other hand, since the relay is away from node \A, it has to spend more power to transmit packets to node \A. Therefore, as the relay \R takes distance from node \A and approaches node \B, more power should be allocated to the relay. This issue is also illustrated in Figure \ref{fig:2}.

\begin{figure}
	\centering
	\subfloat[$d_\A=0.1$]{\includegraphics[width=0.85\linewidth]{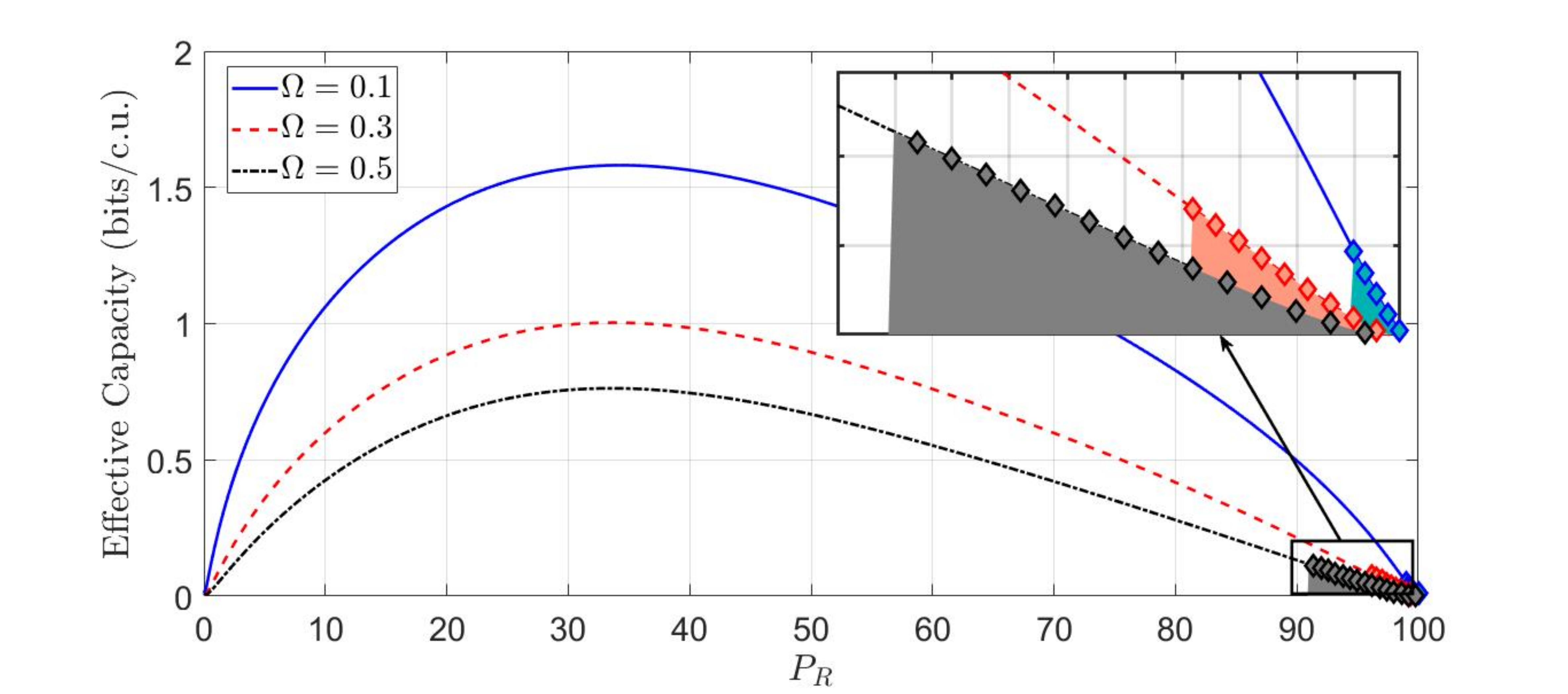}
		\label{fig:3a}}
	\hfil
	\subfloat[$d_\A=0.9$]{\includegraphics[width=0.85\linewidth]{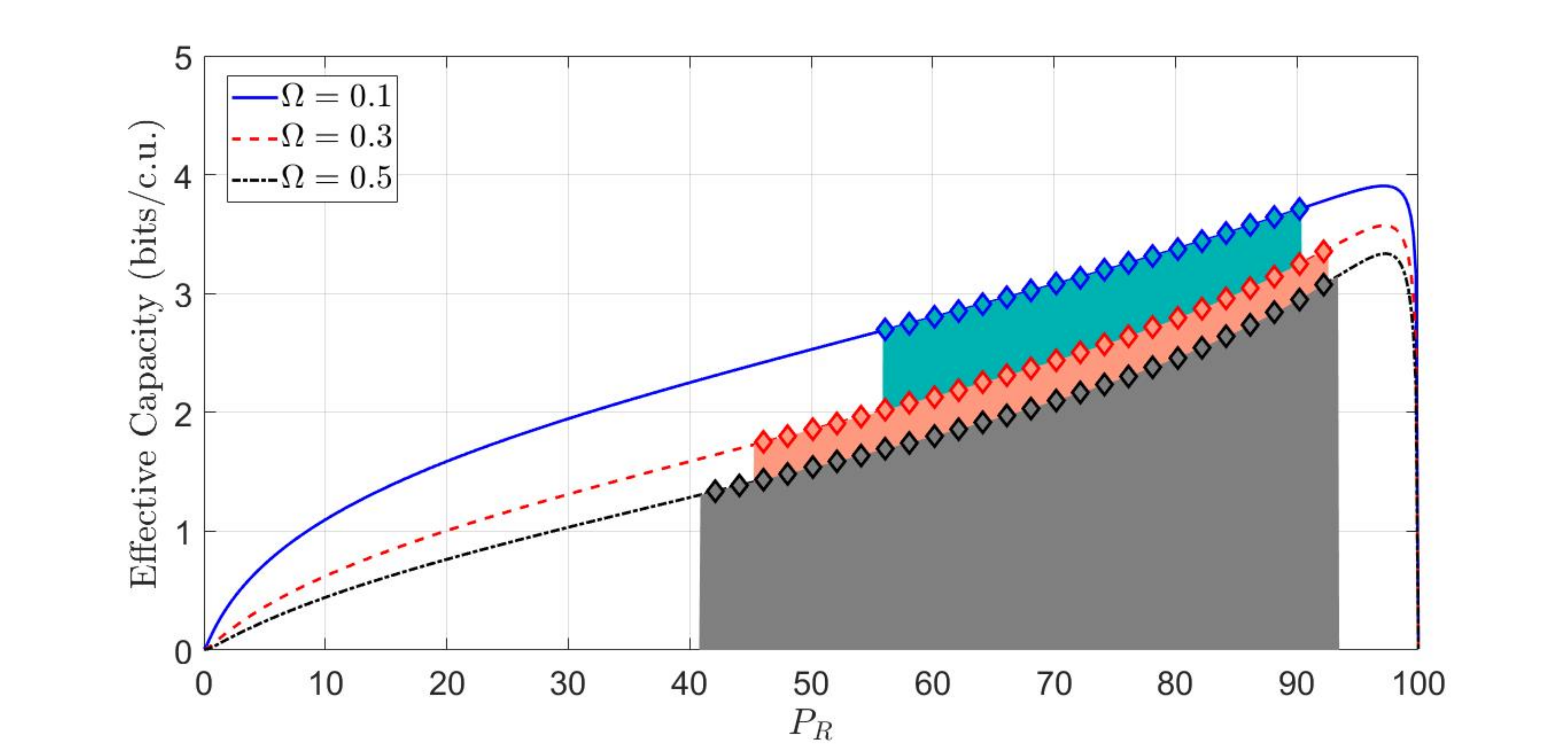}
		\label{fig:3b}}
	\hfil	
	\caption{Effective capacity of node \A vs. FD relay power and different mean residual power interference coefficient (a) $d_{\A}=0.1$, (b) $d_{\A}=0.9$.}
	\label{fig:3}
\end{figure}

Lemma \ref{lem:5} showed that the effective capacity in the FD relay is not necessarily concave relative to the relay power. However, the effective capacity always maximizes at one point. Moreover, by decreasing the mean residual power interference coefficient $\Omega$, the FD SINR tends to the SNR of the HD relay. Figure \ref{fig:3} indicates the effective capacity of node \A in terms of FD relay power for three values of $\Omega=0.1$, $\Omega=0.3$ and $\Omega=0.5$. In Figure \ref{fig:3a} $d_{\A}=0.1$ and in Figure \ref{fig:3b} we have $d_{\A}=0.9$. The effective capacity only has one maximum point in all three cases. However, the effective capacity is not constantly concave and parts of the curve in which the second-order derivative is positive are highlighted in Figure \ref{fig:3}. Moreover, these parts are reduced by decreasing $\Omega$. Since the power interference coefficient is less than 0.1 in real constraints, the effective capacity in FD relay and under practical constraints $\Omega\ll 1$ can be considered concave.

\begin{figure}
	\begin{center}
		\includegraphics[draft=false,width=0.95\linewidth]{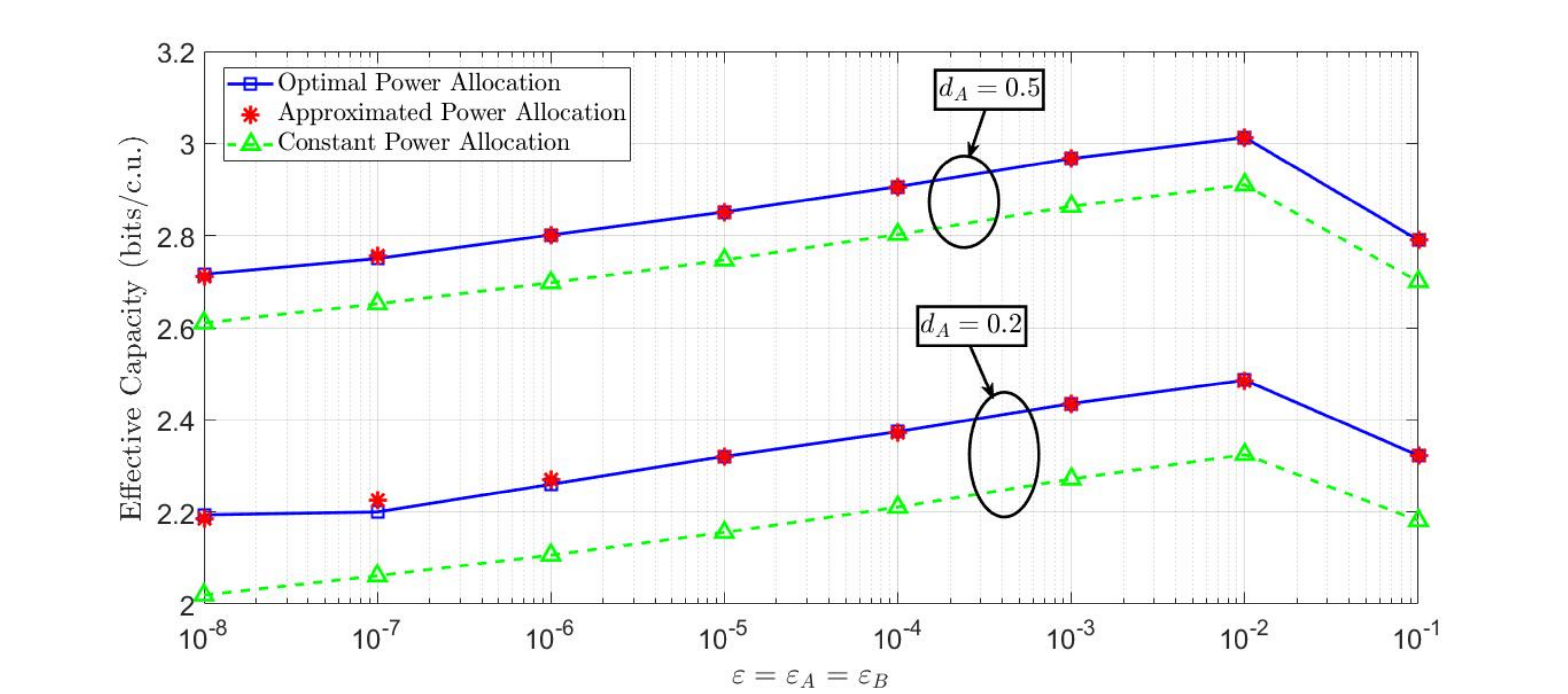}
	\end{center}
	\caption{Weighted sum of the effective capacity vs. error probability in HD relay with optimal and approximated power allocation and constant power allocation.}
	\label{fig:4}
\end{figure}

Figure \ref{fig:4} indicates the weighted sum of the effective capacities for two nodes \A and \B with an identical weight of $w=0.5$ in terms of error probability in HD relay. The effective capacity with approximate power allocation (optimization problem \eqref{eq:HD_min_max_problem_2}) matches with effective capacity with optimal power allocation (optimization problem \eqref{eq:HD_multi_objective_problem_1}) properly. Since solving the simplified optimization problem \eqref{eq:HD_min_max_problem_2} has higher execution speed, we use the approximated solution for further comparison. The execution speeds are compared in Table \ref{tab:2}. Figure \ref{fig:4} presents the weighted sum of the effective capacities of two nodes with equal power allocation $P_{tot}/3$ between nodes \A and \B and relay \R. Accordingly, effective capacity with optimal power allocation provides better performance than equal power allocation between nodes and relays. This superior performance is most noticeable under the constraints that the relay is not placed in the middle of the nodes.

Figure \ref{fig:4} illustrates the weighted sum of the effective capacity of two nodes with an identical weight of $w=0.5$ for $d_{\A}=0.2$ and $d_{\A}=0.5$. The identical weight of $w=0.5$ means that nodes \A and \B have equal importance and the multi-objective optimization problem attempts to maximize the effective capacity of both nodes with the same priority. When the relay is placed between two nodes, the weighted sum of the effective capacity is more than the case that the relay is distant from node B. If the relay is distant from node \B, the attenuation of the relay-node \B channel is more than that of the node \A. Accordingly, part of the available power is consumed to compensate for this additional attenuation so that two nodes with the same priority would become similar in terms of channel attenuation. Then, the residual power is divided to maximize the weighted sum of the effective capacities of nodes \A and \B. Therefore, the performance of the system with the same priority between nodes is greater for the case that the relay is located in the middle of two nodes. Moreover, the effective capacity is maximized at a certain value of the error probability. Therefore, the values of $m$, $\epsilon$ and $\theta$ should be proportional to each other. Similar to the simulations in Figure \ref{fig:4}, other simulations are conducted for the weighted sum of the effective capacity in terms of different lengths of packet $m$ that have similar results. To avoid redundancy, these figures are not provided here.

\begin{figure}
	\begin{center}
		\includegraphics[draft=false,width=0.95\linewidth]{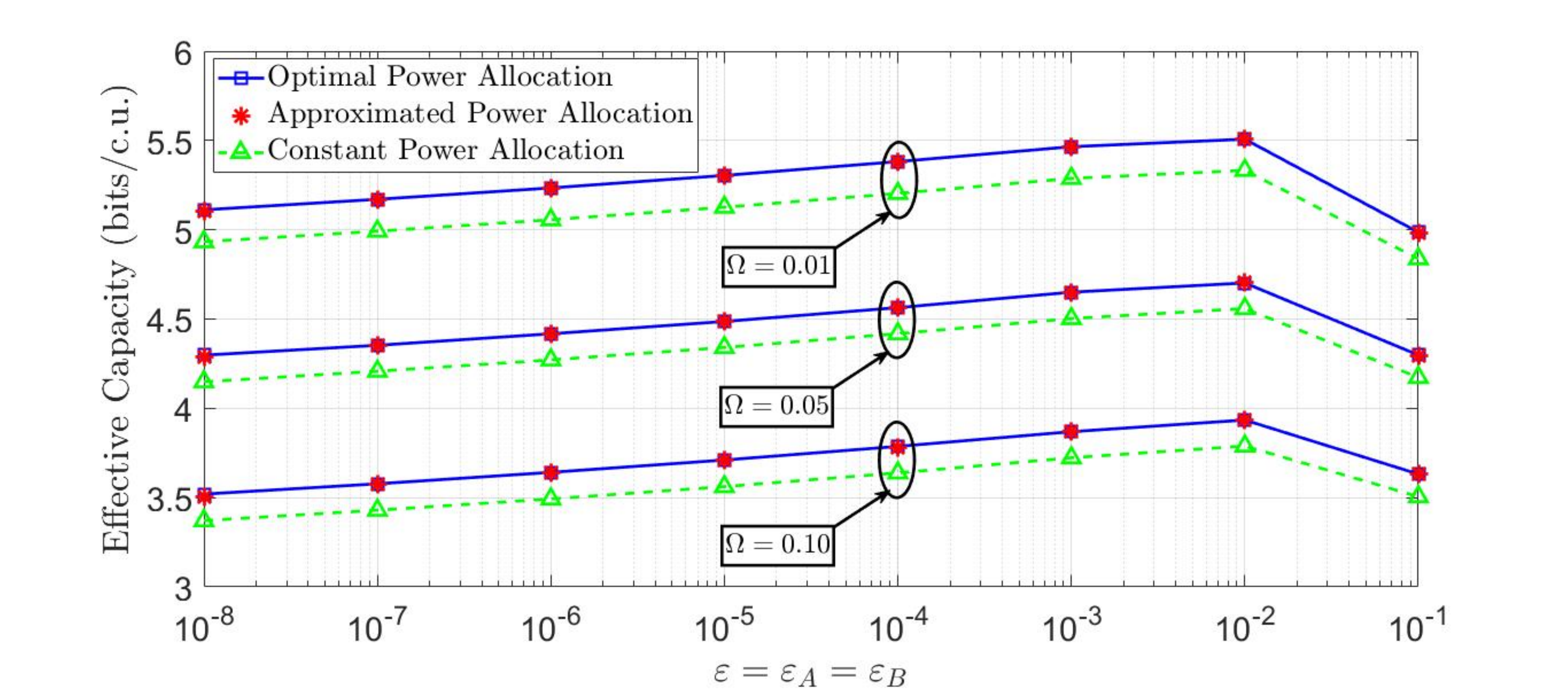}
	\end{center}
	\caption{Weighted sum of the effective capacity vs. error probability in FD relay with optimal and approximated power allocation and constant power allocation.}
	\label{fig:5}
\end{figure}

Figure \ref{fig:5} presents the weighted sum of the effective capacity with the identical weight $w=0.5$ in terms of the error probability $\epsilon=\epsilon_{\A}=\epsilon_{\B}$ in the FD relay between nodes \A and \B. This simulation is conducted for three diferent conditions with almost no self-interference and $\Omega=0.01$, moderate self-interference and $\Omega=0.05$, and high self-interference and $\Omega=0.1$. The weighted sum of effective capacity is matched between the optimal solution of the problem \eqref{eq:FD_multi_objective_problem_2} and the approximate solution of the problem \eqref{eq:FD_min_max_problem_1}, which indicates the high accuracy of the simplified optimization problem. Solving the simplified optimization problem has a much higher execution speed. Therefore, the approximate power allocation will be applied for further comparisons. The execution time to solve the optimal and simplified problem is compared in Table \ref{tab:2}.

Figure \ref{fig:5} also shows the weighted sum of the effective capacity of the two nodes with equal power allocation $P_{tot}/3$. Effective capacity with optimal power allocation enjoys much better performance than the identical power allocation between the nodes and relays. Besides, according to our expect, by reducing the self-interference power coefficient $\Omega$, the effective capacity will increase.

\begin{figure}
	\begin{center}
		\includegraphics[draft=false,width=0.95\linewidth]{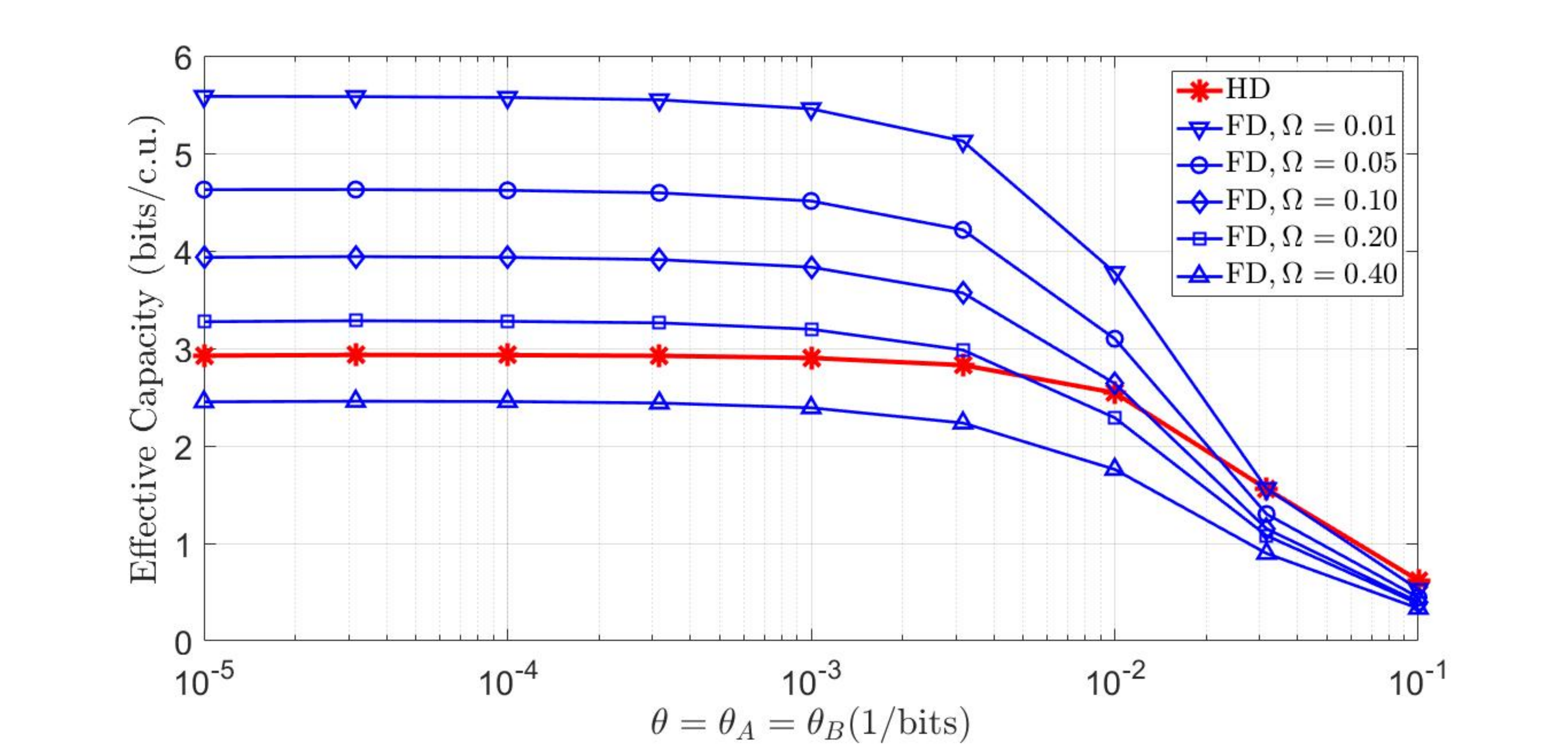}
	\end{center}
	\caption{Comparison of weighted sum of the effective capacity in HD relay and FD relay with different mean residual power interference coefficient $\Omega$.}
	\label{fig:6}
\end{figure}

Fiure \ref{fig:6} provides a comparison between the weighted sum of the effective capacity in the HD and FD relays. The FD relay has self-interference, which affects the overall performance of the system. Accordingly, at low self-interference values (e.g. $\Omega=0.01$), the effective capacity of an FD relay is almost twice greater than that of the HD relay. The reason for this phenomenon is that the FD relay uses only one-time or frequency interval to receive packet and transmit them to the nodes. However, the HD relay uses two-time or frequency intervals to transmit the same packet. When the self-interference is removed properly, the performance of the FD relay is twice that of the HD relay. 

By increasing the self-interference coefficient $\Omega$, the FD relay’s performance decreases. Even at large quantities of $\theta$, which indicates the strict quality of service and very low buffer delay, the HD relay shows a better performance. In other words, self-interference’s effect is more highlighted at large quantities of $\theta$ and strict quality of service. Therefore, if strict service quality is required, it is suggested to remove self-interference in the FD relay completely or use the HD relay. With further increase in $\theta$, the FD relay performance will eventually be worse than the HD relay at all values of $\theta$. In such cases, the self-interference is so severe that it can even compensate the double use of HD relays in time or frequency intervals.

\begin{figure}
	\begin{center}
		\includegraphics[draft=false,width=0.95\linewidth]{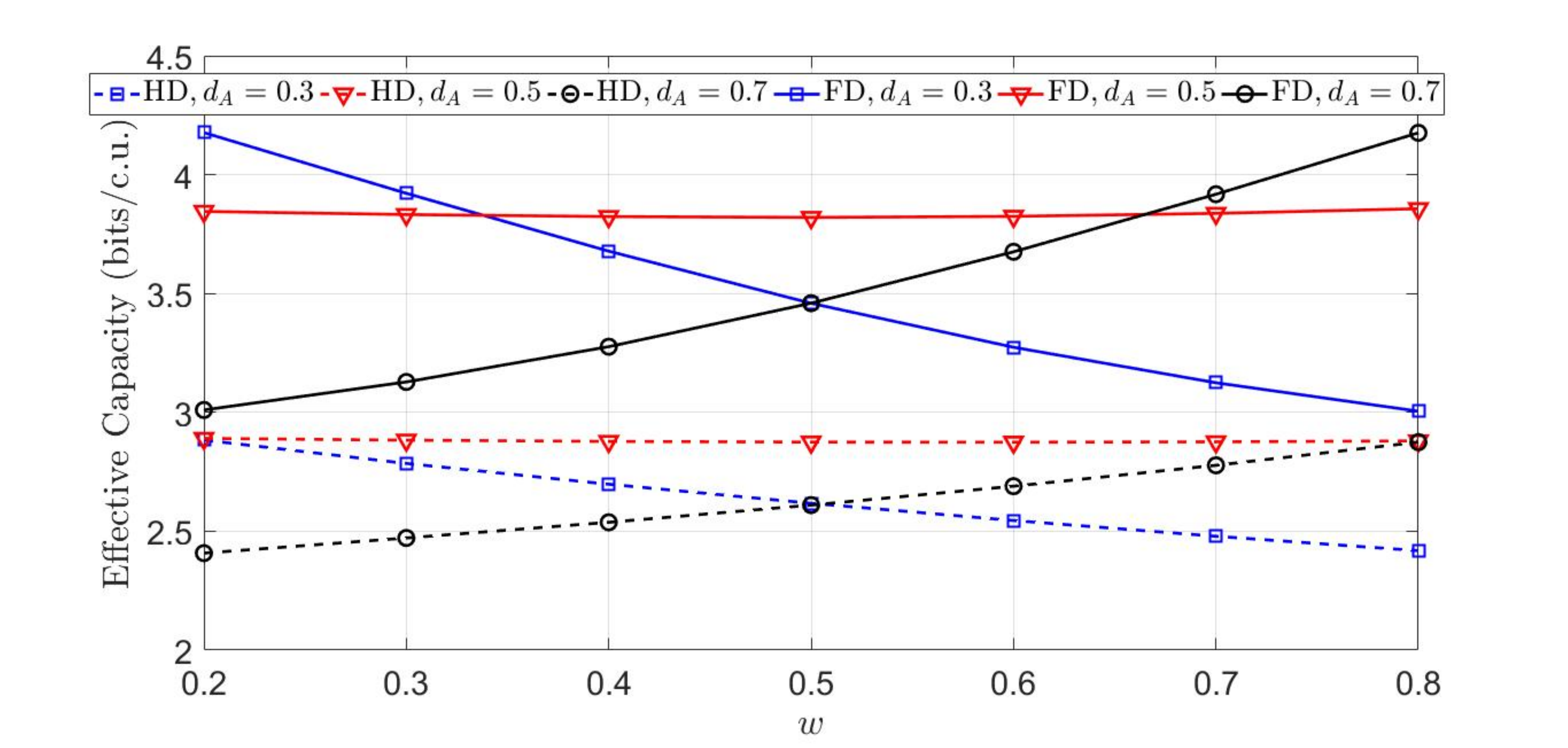}
	\end{center}
	\caption{Comparison of weighted sum of the effective capacity in HD relay and FD relay vs. weight $w$ and different distance between relay and nodes.}
	\label{fig:7}
\end{figure}

In the proposed model, the priority of nodes \A and \B are not necessarily the same. Therefore, in transforming the multi-objective problem and providing the weighted sum of the effective capacity, the weight $w$ might have different values. Figure \ref{fig:7} shows the weighted sum of the effective capacity in terms of $w$ for the HD relay and FD relay with $\Omega=0.1$. To provide a better comparison, the distance of the relay to node \A is set to the three values of $d_{\A}=0.3$, $d_{\A}=0.5$ and $d_{\A}=0.7$. When $d_{\A}=0.5$ and the relay is in the middle of two nodes, the coefficient $w$ does not have much effect on the weighted sum of the capacity. If the relay is in the middle of two nodes, the SNR (and SINR) of the two nodes are symmetric. On the average case, the weighted sum of the effective capacity of the system is unique under all values of $w$. However, when the relay is not in the middle of two nodes, the priority of nodes and the $w$ value have a relatively significant effect on the weighted sum of the effective capacity. However, when the relay is close to node \A and $d_{\A}=0.3$, the weighted sum of the effective capacity decreases with increasing $w$ (increasing the performance priority of node \A). In this case, the relay channel to node \B has higher attenuation than the relay channel to node \A. In lower values of $w$, the power is divided to compensate the poor condition of the relay channel to node \B, and relatively good performance is results. However, as $w$ increases and the priority of node \A is enhances, the power is divided such that node \A achieves an appropriate performance. Therefore, the performance of node \B is ignored. In this case, the inappropriate condition of the relay channel to node \B have an impact on the overall system performance and reduces the weighted sum of the effective capacity. When $d_{\A}=0.7$, the situation is completely reversed and by increasing $w$, the weighted sum of capacity is enhanced. Therefore, to achieve proper performance, when the priority of nodes is different, it is necessary to locate the relay between nodes in a right place and it has a significant effect on the overall system performance.

In the present paper, the problem is solved as multi-objective optimization instead of using a single-objective optimization approach. In the single-objective problem, the objective function is defined as the mean effective capacity of two nodes. Therefore $R_E=(R_{E,\A}+R_{E,\B})/2$ is maximized. The solution to this problem is equivalent to solving the multi-objective problem with $w=0.5$. Therefore, modeling the problem as multi-objective and its solution is more comprehensive than the conventional single-objective model. Therefore, the results should be better. This improvement in performance is quite evident in Figure \ref{fig:7} when the relay is not set in the middle of nodes \A and \B. For example, when the relay is near node \A and $d_{\A}=0.3$, the single-objective solution’s answer is $R_E=3.46$(bits/c.u.) in FD mode ($R_E=2.61$(bits/c.u.)  in the HD mode). However, in the multi-objective solution, if the weight $w$ is chosen properly, the weighted sum of the effective capacity may be higher than the mean effective capacity by about 25 percent in the FD mode (15 percent in HD mode). Another point in multi-objective optimization is the possibility of prioritizing the performance of nodes, which improves the overall system performance. 

\begin{figure}
	\begin{center}
		\includegraphics[draft=false,width=0.95\linewidth]{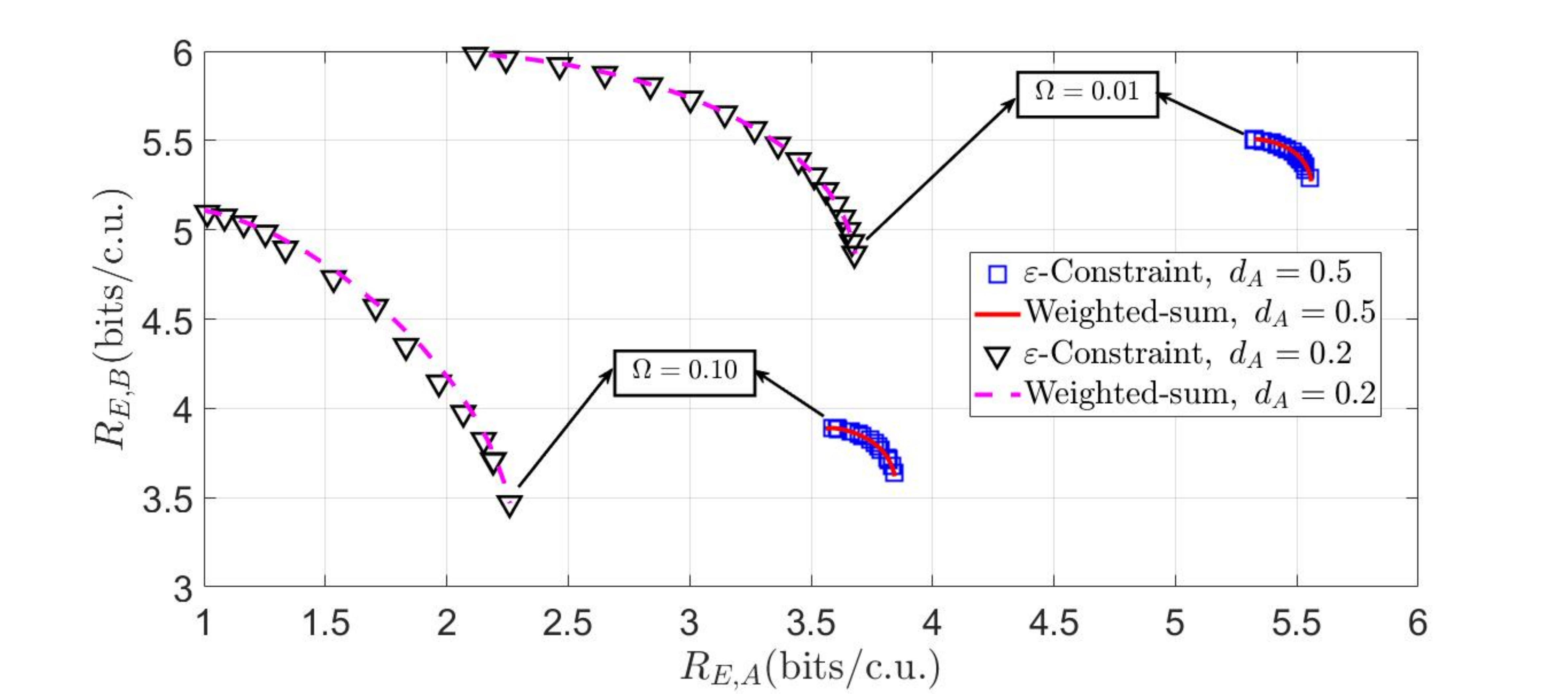}
	\end{center}
	\caption{Pareto optimal frontier with the weighted sum and $\epsilon$-constraint methods.}
	\label{fig:8}
\end{figure}

It was discussed in Section \ref{subsec:multi_objective_FD} that the weighted sum method is the only sufficient condition to reach the optimal Pareto point in the non-concave multi-objective problems. To evaluate the performance of this method, the calculated Pareto point by the weighted sum method and $\epsilon$-constraint method are plotted in Figure \ref{fig:8}. The $\epsilon$-constraint method is the necessary and sufficient condition to reach all the optimal Pareto points in non-concave problems \cite{p26}. This method is very simple when there are only two objective functions in the multi-objective optimization problem. Accordingly, one of the objective functions such as $R_{E,\B}$ in \eqref{eq:FD_multi_objective_problem_1} is transferred to the problem constraints and the optimization problem is rewritten as:
\begin{equation}\label{eq:36}
\left\{ \begin{array}{l}
\mathop {\max }\limits_{{P_R}} {\rm{   }}{R_{E,{\A}}}\\
{\rm{s}}{\rm{.t}}{\rm{.   }}{R_{E,{\B}}} > \mu \\
{\rm{s}}{\rm{.t}}{\rm{.   }}{P_{\R}} < {P_{tot}}
\end{array} \right.
\end{equation}
In this problem, the Pareto points is calculated by altering $\mu$ and solving the problem (in the $\epsilon$-constraint method, when one of the objective functions is added to the problem constraints, this function is compared with $\epsilon$ and this is the reason of naming this method. However, in this article, $\epsilon$ is used elsewhere; therefore $R_{E,\B}$ is compared with $\mu$ in the added constraint). 

In Figure \ref{fig:8}, the optimal Pareto frontier is plotted by the weighted sum and $\epsilon$-constraint in FD relay. For a detailed analysis, the distance between the relay and node \A is assumed as $d_{\A}=0.5$ and $d_{\A}=0.2$ individually. Moreover, the intensity of self-interference is considered as the low and relatively high values of $\Omega=0.01$ and $\Omega=0.10$ respectively. The calculated Pareto frontier by both methods are almost identical. Therefore, the solution to the non-concave optimization problem \eqref{eq:FD_multi_objective_problem_1} with the weighted sum method is very close to the optimal performance.

If the relay is in the middle of two nodes and $d_{\A}=0.5$, the performance is symmetric and the range of variations for two nodes is the same. In addition, the range of variations is not extensive. However, when the relay is close to node \A,  symmetry is lost and the variation interval increases significantly. For example, when $\Omega=0.01$, the effective capacity of node \A varies within the range of $R_{E,\A}=2.1$(bits/c.u.) to $R_{E,\A}=3.7$(bits/c.u.) and that of node \B varies between $R_{E,\B}=4.9$(bits/c.u.) and $R_{E,\B}=6.0$(bits/c.u.). If $\Omega=0.10$, the effective capacity of node \A varies within the range of $R_{E,\A}=1.0$(bits/c.u.) to $R_{E,\A}=2.2$(bits/c.u.) and that of node \B varies between $R_{E,\B}=3.5$(bits/c.u.) and $R_{E,\B}=5.1$(bits/c.u.). Therefore, especially when the relay is not located between the two nodes \A and \B, power allocation is very efficient, which can significantly affect the overall system performance.
\begin{table}
  \centering
  \caption{Comparison of execution time in the optimal and approximated problem.}\label{tab:2}
  \begin{tabular}{cccc}
\hline
\multicolumn{1}{|c||}{relay mode}            & \multicolumn{3}{c|}{HD}                                                           \\ \hline
\multicolumn{1}{|c||}{$\epsilon=\epsilon_{\A}=\epsilon_{\B}$}                     & \multicolumn{1}{c|}{$10^{-8}$}   & \multicolumn{1}{c|}{$10^{-5}$}   & \multicolumn{1}{c|}{$10^{-2}$}   \\ \hline
\multicolumn{1}{|c||}{optimal solution}      & \multicolumn{1}{c|}{1415} & \multicolumn{1}{c|}{1356} & \multicolumn{1}{c|}{1327} \\ \hline
\multicolumn{1}{|c||}{approximated solution} & \multicolumn{1}{c|}{564}  & \multicolumn{1}{c|}{561}  & \multicolumn{1}{c|}{561}  \\ \hline
                                            &                           &                           &                           \\ \hline
\multicolumn{1}{|c||}{relay mode}            & \multicolumn{3}{c|}{FD}                                                           \\ \hline
\multicolumn{1}{|c||}{$\Omega$}                     & \multicolumn{1}{c|}{$0.01$}   & \multicolumn{1}{c|}{$0.05$}   & \multicolumn{1}{c|}{$0.10$}   \\ \hline
\multicolumn{1}{|c||}{optimal solution}      & \multicolumn{1}{c|}{1935} & \multicolumn{1}{c|}{1657} & \multicolumn{1}{c|}{1622} \\ \hline
\multicolumn{1}{|c||}{approximated solution} & \multicolumn{1}{c|}{553}  & \multicolumn{1}{c|}{540}  & \multicolumn{1}{c|}{547}  \\ \hline
\end{tabular}
\end{table}

Table \ref{tab:2} provides a comparison between the execution time of optimization problem in both optimal and approximated problem for HD relay with 3 different error probability and FD relay with 3 different values of $\Omega$. The simulations were implemented on a laptop with an 8-core Intel (R) Core (TM) i7-3632QM CPU @ 2.20GHz processor with 6GB of memory. For any level of error probability, 1000 random samples are generated as a channel and simulations are performed to achieve the gradient value of less than 0.1 at each point. Then the  required time for optimal and approximated solution are recorded. This procedure is repeated 100 times and the results of the average runtime are presented in Table \ref{tab:2}. The execution time of the approximated problem is much shorter than that of the optimal problem. This difference increases with the increasing number of channel samples.


\section{Conclusions}\label{sec:conclusions}
Considering the significance of low-latency data transmission in the fifth generation of communications, this article addressed transmission of the short packets between two nodes. Short packet data transmission requires short transmission time and decoding. Therefore, it is appropriate to delay-sensitive traffics. Moreover, the two-way relay is applied for data transmission between two nodes. The relay type could be either HD or FD. If an FD relay is applied, the data transmission time between two nodes is reduced. Finally, the weighted sum of the effective capacity for two nodes is considered to evaluate the performance. The effective capacity indicates the transferable rate through the channel by guaranteeing the statistical delay of packet in the transmitter buffer.

In the proposed model, the optimal power for effective capacity maximization of two nodes is allocated when the packets are transferred between these two nodes with a two-way HD or FD relay. It is also proved that the resulting performance is better than identical power allocation among the nodes. Since the performance priority of the two nodes is not necessarily the same, the problem optimization is considered multi-objective. In the case of using the HD relay, it is proved that the optimization problem is concave and it was solved by the weighted sum method. However, in the case of using the FD relay, although the problem has one optimal point, it is not concave. Therefore, the weighted sum method was also applied to solve the optimization problem in these conditions. As the optimization problem involves statistical means, its solution is long and time-consuming. Thus, the approximate method was obtained to calculate the optimal power of the relays and nodes. The answer of solving the approximate problem has an appropriate consistency with the results of the optimal solution. It is also proved that multi-objective optimization provides better results than solving the problem with the single-objective approach. In the single-objective case, the effective capacity of two nodes with similar weights is summed up and maximized. However, in the multi-objective case, the priority and weight of the nodes are not necessarily equal. Therefore, the overall system performance is better than the performance of the system in single-mode especially when the relay is not in the middle of two nodes or their priority is not the same. Regarding the performance of HD and FD relays, the FD relay outperforms the HD relay when the self-interference is reduced to an appropriate level. With increasing $\theta$ and the strict service requirement of the nodes, the performance of both HD and FD relays decreases naturally. However, in this case, HD outperforms FD gradually. Therefore, it is concluded that if the strict quality of service is considered, it is advised to remove self-interference in the FD relay largely or use the HD relay.

\bibliographystyle{unsrt}
\bibliography{References}

\end{document}